\documentclass{article}
\usepackage{authblk}

\usepackage{tikz}
\usetikzlibrary{shapes, calc}

\usepackage[a4paper, total={6in, 8in}]{geometry}

\usepackage{amsmath}
\usepackage{amssymb}
\usepackage{amsthm}
\newtheorem{theorem}{Theorem}
\newtheorem{lemma}{Lemma}
\newtheorem{example}{Example}

\usepackage{comment}
\usepackage{graphicx}
\usepackage{float}
\usepackage{subcaption}
\newcommand{\Q}{\mathbb{Q}}

\newcommand{\cbold}{\mathbf{c}}

\newcommand{\catw}[1]  {
  \mathbf{#1}
}
\newcommand{\mor}[3]  {
  #1 \colon #2 \rightarrow #3
}
\newcommand{\tuple}[1]  {
	\langle #1 \rangle
}
\newcommand{\boolalg}[1]  {
  \catw{Bool}_{\catw{Set}[1]}
}

\newcommand{\pprob}[1]  {
	#1_n[x]
}

\newtheorem{remark}{Remark}

\usepackage{algorithm}
\usepackage{algpseudocode}

\usepackage{listings}
\usepackage{xcolor}

\definecolor{codegreen}{rgb}{0,0.6,0}
\definecolor{codegray}{rgb}{0.5,0.5,0.5}
\definecolor{codepurple}{rgb}{0.58,0,0.82}
\definecolor{backcolour}{rgb}{0.96,0.96,0.96} 

\lstdefinestyle{pythonstyle}{
    backgroundcolor=\color{backcolour},   
    commentstyle=\color{codegreen},
    keywordstyle=\color{magenta},
    numberstyle=\tiny\color{codegray},
    stringstyle=\color{codepurple},
    basicstyle=\ttfamily\footnotesize, 
    breakatwhitespace=false,         
    breaklines=true,                 
    captionpos=t,                    
    keepspaces=true,                 
    numbers=left,                    
    numbersep=5pt,                  
    showspaces=false,                
    showstringspaces=false,
    showtabs=false,                  
    tabsize=4,
    frame=lines,                     
    rulecolor=\color{black}
}

\title{Information-Based Complexity vs Computational Complexity in Phaseless Polynomial Interpolation}

\author[1]{Michał R. Przybyłek}
\affil[1]{Polish-Japanese Academy of Information Technology; Warsaw, Poland\protect\\ \texttt{mrp@mimuw.edu.pl}}

\author[2]{Paweł Siedlecki}
\affil[2]{University of Warsaw; Warsaw, Poland\protect\\ \texttt{psiedlecki@mimuw.edu.pl}}

\date{October 25, 2025}

\begin{document}

\maketitle
\abstract{The authors of \cite{przybylek2020note} have shown that phaseless polynomial interpolation over $\mathbf{Q}$ is possible with $n+2$ points, where $n$ is the upper-bound on the degree of a polynomial. Nonetheless, their reconstruction algorithm and the method of adaptively choosing evaluation points are exponential time. On the other hand, they have also shown that given $2n+1$ points, the polynomial can be reconstructed in a polynomial time. In \cite{przybylek2020note} a conjecture have been put forward, namely 
that the reconstruction problem from such $n+2$ points is exponential time. Moreover, a question about the number of points sufficient for polynomial time reconstruction have been posed. In this paper, we answer these questions -- we show that (1) reconstruction problem from $2n-k$ for any constant $k$ is polynomial time, (2) reconstruction problem from $(1+c)n+2$ points for any constant $c \in [0, 1)$ is NP-Complete, (3) evaluation points admitting a unique solution can be chosen in polynomial time.}

\section{Introduction}

The challenge of phase retrieval--the reconstruction of a signal, function, or vector from the magnitude of its measurements--is a foundational problem in applied mathematics \cite{jaganathan2016phase}, physics, and engineering. It appears in diverse fields such as X-ray crystallography, microscopy, optics, and signal processing. In the algebraic context, this problem manifests as phaseless polynomial interpolation: the reconstruction of a polynomial $p(x)$ from a set of evaluation points $x_i$ and their corresponding absolute values, $|p(x_i)|$.

This problem immediately highlights a fundamental dichotomy, which is the central theme of this work: the distinction between \emph{information-based complexity} and \emph{computational complexity}. The first asks: how many sample points are \emph{informationally sufficient} to uniquely determine the polynomial (up to an unobservable global phase)? The second asks: given a sufficient number of points, what is the \emph{computational tractability} of an algorithm that performs this reconstruction? A problem may be solvable in principle (i.e., have sufficient information) but intractable in practice (i.e., require exponential time).

Recent work, notably \cite{przybylek2020note}, has precisely framed this gap. For polynomials of degree at most $n$ with rational coefficients ($\mathbf{Q}[x]$), it has been shown that $n+2$ phaseless evaluation points are informationally sufficient and necessary for unique reconstruction. This establishes a low information-theoretic bound. In contrast, the authors of that paper provided a polynomial-time reconstruction algorithm when $2n+1$ points are given. The algorithm they provided for the $n+2$ point case, however, requires exponential time.

This disparity left two critical questions unanswered, which that paper posed as open problems:
\begin{enumerate}
    \item Is the reconstruction problem from $n+2$ points \emph{inherently} of exponential-time complexity?
    \item What is the true computational threshold? What is the minimum number of points required for a \emph{polynomial-time} reconstruction?
\end{enumerate}

In this paper, we resolve both of these questions. We provide a complete characterization of the computational complexity landscape for this problem, drawing a sharp, and perhaps surprising, line between the tractable and the intractable.

Our first main result addresses the polynomial-time regime. We demonstrate that the $2n+1$ point requirement is not optimal. We prove that reconstruction from $m = 2n-k$ points, for any fixed constant $k$, is solvable in \emph{polynomial time} (in $n$ and the bit-size of the input). Our method relies on techniques from computational algebraic geometry. We first parameterize the $(k+1)$-dimensional affine space of all degree $\le 2n$ polynomials $p$ satisfying the $2n-k$ constraints (where the constraints are $p(x_i) = |q(x_i)|^2$). We then impose the algebraic condition that the solution must be a perfect square. This results in a system of polynomial equations in $k+1$ variables. Since $k$ is constant, this system can be solved in polynomial time using Gröbner basis methods augmented with LLL-reductions.

Our second main result answers the intractability conjecture. We prove that the reconstruction problem from $(1+c)n+2$ points, for any constant $c \in [0, 1)$, is \emph{NP-complete}. This includes the $n+2$ point case. We prove this by a reduction from the Partition Problem. We show that finding the correct set of signs $c_i \in \{-1, 1\}$ for the evaluations $p(x_i) = c_i |p(x_i)|$ is equivalent to solving an instance of the Partition Problem. This establishes that the exponential barrier observed in prior work is not an algorithmic artifact but an inherent feature of the problem's complexity.

Finally, we also show that in case of less than $2n+1$ evaluation points adaptation is necessary, but it is sufficient to use only a single evaluation point adaptively and the choice can be done in polynomial-time.

The paper is structured as follows. In Section 2 we provide necessary definitions and notions. In Section 3 we present the polynomial-time algorithm for $2n-k$ points, detailing the parameterization and the algebraic-geometric solution. In Section 4, we present our NP-completeness proof.

\section{Setting}
The phaseless polynomial interpolation problem over a field $\mathbf{K}$, where $\mathbf{K}$ is a subfield of the field 
$\mathbf{C}$ of complex numbers, can be  described informally as follows. Given an upper bound $n$ on the degree of polynomials $q \in \mathbf{K}[x]$, choose numbers $x_i \in \mathbf{K}$ for $0 \leq i \leq N$ such that given nonnegative values $y_i\in\mathbf{K}$ there is a unique, up to an absolute value, polynomial $q \in \mathbf{K}_n[x]$ interpolating $y_i$ without a phase, i.e. we have $|q(x_i)| = y_i$ and if a polynomial $q'$ of degree at most $n$ satisfies $|q'(x_i)| = y_i$, then $q'=\alpha q$ for some $\alpha\in\mathbf{K}$ with $|\alpha|=1$. There are two versions of performing the above choice: non-interactive and interactive. A non-interactive choice is when the interpolation procedure has to choose \emph{all} numbers $x_i \in \mathbf{K}$ for $0 \leq i \leq N$ first and then the corresponding values $y_i \in \mathbf{K}$ are revealed. An interactive choice is when the interpolation procedure has to choose a single $x_i$ at a time, then the corresponding value $y_i$ is revealed and the procedure decides whether or not it needs more values and if so it chooses $x_{i+1}$ based on revealed values $y_k$ for $0 \leq k \leq i$.

Moreover, in this paper, we are interested not only in the choice of points $x_i$ (interactively or not), but also in the \emph{effectiveness} of the process of reconstruction of the interpolating polynomial $q \in \mathbf{K}_n[x]$, 
which culminates in listing its coefficients $a_j \in \mathbf{K}$ for $0 \leq j \leq n$.

Although the reconstruction problem for the non-interactive choice of interpolation points can be formalized as a classical problem in computational complexity, that is: ``given $\tuple{x_i, y_i}$ for the input find phaseless interpolating polynomial $q \in \mathbf{K}_n[x]$'', the formalisation of the interactive choice is more challenging. It, clearly, cannot be formalised as a problem in computational complexity, because there is no single well-defined input to the problem--the algorithm itself can control (to some degree) what input will be provided to the algorithm next, or if it is satisfied with the input given so far. This is less of a problem if we want to provide positive results--to show that the problem can be solved in, say, polynomial time, it suffices to show an interactive procedure that runs in polynomial time. 
Proving negative results is more challenging, especially if we believe that the problem is related to complexity classes like ``non-deterministic polynomial time''. The reason is twofold: the classical machinery of stating that a problem is hard for a given class cannot work, because the interactive choice as stated is not an algorithmic problem in the classical sense; on the other hand, restricting the instances to the problems with unique (up to a phase) solutions gives a theoretical barrier in proving that a problem is hard for ``non-deterministic polynomial time''' class (i.e.~it is a difficult open question whether or not there are any NPH problems with unique solutions). For this reason, we can prove the hardness of the problem for non-interactive choice only.

The classical setting of the problem, as described in \cite{przybylek2020note}, is given by the framework of Information-Based Complexity (IBC). According to this setting, a problem consists of a function
$$\mor{P}{V}{W}$$
between a set $V$ and (usually) a metric space $W$ and a class $\Lambda$ of basic information operations (or basic measurements)
$$V \rightarrow \mathbf{K}$$
An approximate solution to a problem $P, \Lambda$ consists of an \emph{information operator} 
$N \colon V\to\mathbf{K}^*$ and an \emph{algorithm} 
$\phi\colon\mathbf{K}^*\to V$ (with $\mathbf{K}^*$ denoting the set of all finite sequences of elements from $\mathbf{K}$) such that: $\phi\circ N \approx P$ like on Figure~\ref{fig:approx_comm_small}. 

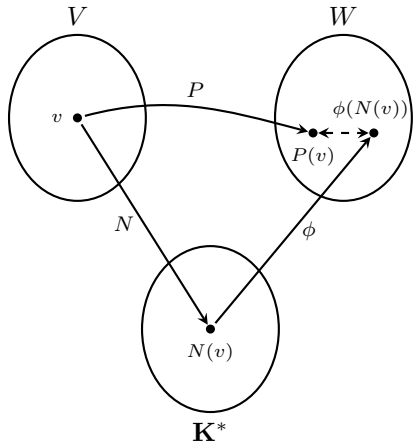
\begin{figure}[h]
    \centering
    \begin{tikzpicture}[
        >=stealth,
        set_node/.style={draw, ellipse, minimum width=1.8cm, minimum height=2.2cm, thick},
        element_node/.style={circle, fill, inner sep=1.2pt},
        func_label/.style={midway, font=\footnotesize}
    ]

        \node[set_node, label=above:$V$] (V) at (0, 0) {};
        \node[set_node, label=above:$W$] (W) at (3.5, 0) {};
        \node[set_node, label=below:$\mathbf{K}^*$] (K) at (1.75, -2.8) {};

        \node[element_node, label=left:\scriptsize $v$] (v) at (V.center) {};
        \node[element_node, label=below:\scriptsize $N(v)$] (Nv) at (K.center) {};

        \node[element_node, label=below:\scriptsize $P(v)$] (Pv) at ($(W.center) - (0.4, 0.2)$) {};
        \node[element_node, label=above:\scriptsize $\phi(N(v))\;$] (SNv) at ($(W.center) + (0.4, -0.2)$) {};

        \draw[->, thick, shorten <=1pt] (v) to[out=15, in=165] node[func_label, above] {$P$} (Pv);

        \draw[->, thick, shorten <=1pt] (v) -- node[func_label, left] {$N$} (Nv);

        \draw[->, thick, shorten <=1pt] (Nv) -- node[func_label, right] {$\phi$} (SNv);

        \draw[<->, dashed, thick] (Pv) -- node[midway, right, font=\tiny, xshift=1pt] {} (SNv);

    \end{tikzpicture}
    \caption{Information-Based Complexity perspective: problem $P$, information operator $N$ and algorithm $\phi$ such that for every $v\in V$ we have that: $\phi(N(v)) \sim P(v)$.}
    \label{fig:approx_comm_small}
\end{figure}

Moreover, information operator $N$ must be built from basic information operators $\lambda_i \in \Lambda$ in a certain way. In case $N = \tuple{\lambda_i}_{i=0}^j$ for some fixed $j$, we say that the solution is \emph{non-adaptive}, which roughly corresponds to the non-adaptive choice from the informal statement of phaseless interpolation problem. The other case is when the solution is \emph{adaptive}--there is a procedure $\mor{\phi}{\mathbf{K}^*}{\Lambda \sqcup \{\bot\}}$ such that: 
$$N(v)_j = 
\begin{cases} 
    \phi(\epsilon) & \text{if } j = 0 \\
    \phi\big(N(v)_0, N(v)_1, \dots, N(v)_{j-1}\big)(v) & \text{if } j > 0
\end{cases}$$
where the second branch is defined for $i < j$ where $N_j(v) = \bot$, which indicates that no more information will be used by the algorithm. This roughly corresponds to the interactive choice of points from the informal statement of phaseless interpolation problem. 

We usually measure the quality of approximation in terms of the distance 
$$\max_{v\in V}(d_W(\phi(N(v)), P(v)),$$ 
the information complexity in terms of $\max_{v\in V}(|N(v)|)$, and computational complexity in terms of the maximal number of steps performed by a machine realizing $\phi$ (in a chosen computational model) on input $v \in V$, \emph{plus} the number of steps performed by a machine realizing $N$ (again, in some chosen computational model) 
in order to compute basic information operators $\lambda_i$. Often, when it is important which basic information operations from $\Lambda$ a given algorithm uses, instead of providing number $\max_{v\in V}(|N(v)|)$ we give an explicit characterisation of the basic information operations used. More on information operators and their role in IBC in more general settings can  be found, e.g., in \cite{tww1988}, \cite{plaskota1996noisy} and \cite{avi}. The phaseless polynomial interpolation fits into this framework as follows. We set $V = \mathbf{K}_n[x]$, $W = V/\sim$ equipped with the discrete metric, where $q \sim q'$ iff $q = \alpha q'$ for some $\alpha\in\mathbf{K}$ with $|\alpha|=1$, $P(q) = |q|$ and $\Lambda = \{ \delta_x \colon x \in \mathbf{K} \}$, where $\delta_x(q) = |q(x)|$ are absolute-value evaluations. The solution must be exact, i.e., $P(v) = \phi(N(v))$. 

We note in passing that for $q,q'\in\mathbf{K}$ the relation $q \sim q'$ holds if and only if $|q|=|q'|$ as functions 
$\mathbf{K}\to\mathbf{K}$.

Because we are interested in computational complexity in the traditional sense of Turing Machines, we shall restrict our considerations to the field of rational numbers, i.e., to $\mathbf{K} = \mathbf{Q}$. Moreover, we shall assume that all objects under considerations are encoded in an effective way. This assumption provides another natural setting for the polynomial interpolation problem, more familiar to the field of computer science, namely, Query Complexity or Decision Tree Complexity \cite{buhrman2002complexity}.

In the Query Complexity framework, the role of basic information operators $\Lambda$ is played by a black-box oracle $O$ that, if given some input $v$, can produce an output $O(v)$. An algorithm for a problem is an algorithm in the traditional sense relative to the given oracle $O$. We measure the query complexity of the problem in terms of the number of queries performed to oracle $O$ and the computational complexity in the usual way: by the number of steps performed by the algorithm. This setting explicitly deals with the computational complexity necessary for the choice of the query points (e.g.~interpolation points in our problem). The phaseless polynomial interpolation naturally fits into this framework by using an oracle $O$ that encodes absolute value $|p|$ of a given polynomial $p$ -- i.e.~oracle $O$ accepts as an input a rational number $x$ and outputs $|p(x)|$. The task is to identify what polynomial is represented by oracle $O$.

We use the following notational conventions throughout the paper. As indicated above, by 
$$\pprob{\mathbf{K}}=\{a_0+a_1x+\cdots+a_nx^n\colon \ a_0,a_1,\ldots,a_n\in\mathbf{K}\}$$ 
we will denote 
the linear space of all polynomials of degree at most $n$ ($n\in\mathbb{N}$) with coefficients from $\mathbf{K}$.  More generally, if $\tuple{\overline{x}, \overline{y}} = \tuple{x_0, y_0}, \tuple{x_1, y_1}, \tuple{x_2, y_2}, \dots, \tuple{x_{k-1}, y_{k-1}}$ are some interpolation points in $\mathbf{K}^2$, then by $\mathbf{K}_n(\overline{x}, \overline{y})[x]$ we will denote the set of polynomials of degree at most $n$ passing through points $\tuple{x_i, y_i}$ for $0 \leq i < k$. Observe that if $p, q \in \mathbf{K}_n(\overline{x}, \overline{y})[x]$, then $(p-q)(x_i) = p(x_i)-q(x_i) = y_i-y_i=0$, therefore $p-q$ vanishes on $\{x_0, x_1, \dotsc, x_{k-1}\}$. Moreover, in the other direction, if $z \in \mathbf{K}_{n}[x]$ vanishes on $\{x_0, x_1, \dotsc, x_{k-1}\}$, then $p+z$ interpolates $\tuple{x_0, y_0}, \tuple{x_1, y_1}, \dotsc, \tuple{x_{k-1}, y_{k-1}}$, therefore $p+z \in \mathbf{K}_n(\overline{x}, \overline{y})[x]$. Because the set of all polynomials of degree at most $n$ that vanish on $\{x_0, x_1, \dotsc, x_{k-1}\}$ forms an $(n-k)$-dimensional vector space, the set $\mathbf{K}_n(\overline{x}, \overline{y})[x]$ is an $(n-k)$-dimensional \emph{affine} space.

\begin{lemma}\label{l:affine}
    Let $k \in \mathbb{N}$ be a fixed constant. Given $n \geq k$ and a sequence 
    $$\tuple{x_0, y_0}, \tuple{x_1, y_1}, \dotsc, \tuple{x_{n-k}, y_{n-k}}$$ 
    where $x_i$ are pairwise distinct, the set of all polynomials $p(x) \in \mathbf{K}_{n}[x]$ such that $p(x_i) = y_i$ for all $i=0, \dots, n-k$ forms a $k$-dimensional affine space $\mathbf{K}_{n}(\overline{x}, \overline{y})[x]$. The isomorphism $\mathbf{K}^k \approx \mathbf{K}_{n}(\overline{x}, \overline{y})[x]$ is given as:
    $$\cbold \in \mathbf{K}^k \mapsto  p(x, \cbold) = L(x) + \left( \sum_{j=0}^{k-1} c_j x^j \right) \prod_{i=0}^{n-k} (x-x_i)$$
    where $\cbold \in \mathbf{K}^k$ are the coordinates and $L(x)$ is the unique Lagrange interpolating polynomial of degree at most $n-k$ that interpolates $\langle x_i, y_i\rangle_{i = 0}^{n-k}$:
   $$L(x) = \sum_{j=0}^{n-k} y_j \ell_j(x)$$
    with Lagrange basis polynomials $\ell_j(x)$:
    $$\ell_j(x) = \prod_{\substack{i=0 \\ i \neq j}}^{n-k} \frac{x - x_i}{x_j - x_i}$$
\end{lemma}
\begin{proof}
    The argument given above shows that the set forms a $k$-dimensional affine space. Because $p(x; \cbold)$ is clearly a $k$-dimensional affine space, it suffices to show that for every vector $[c_0, c_1, \dots, c_{k-1}]$ the polynomial $L(x) + \left( \sum_{j=0}^{k-1} c_j x^j \right) \prod_{i=0}^{n-k} (x-x_i)$ interpolates $\langle x_i, y_i\rangle_{i = 0}^{n-k}$. But this is obvious, because $L(x)$ interpolates $\langle x_i, y_i\rangle_{i = 0}^{n-k}$ by construction and 
    $$\left( \sum_{j=0}^{k-1} c_j x^j \right) \prod_{i=0}^{n-k} (x-x_i)$$ 
    vanishes on  $\{x_0, x_1, \dotsc, x_{n-k}\}$.
\end{proof}

We shall also write $\mathbf{K}[x] = \bigcup_n \mathbf{K}_n[x]$ for the linear space of all polynomials in variable $x$ and, similarly $$\mathbf{K}(\overline{x}, \overline{y})[x] = \bigcup_{n\geq k}\mathbf{K}(\overline{x}, \overline{y})_n[x]$$ for the affine space of polynomials in variable $x$ interpolating 
$$\tuple{\overline{x}, \overline{y}} = \tuple{x_0, y_0}, \tuple{x_1, y_1}, 
\tuple{x_2, y_2}, \dots, \tuple{x_{k-1}, y_{k-1}}$$

We denote by $\mathbf{K}_{n}^2[x]$ the set of polynomials of degree at most $2n$ that are squares of polynomials in $\mathbf{K}_{n}[x]$. That is:
$$ \mathbf{K}_{n}^2[x] = \{ p(x) \in \mathbf{K}_{2n}[x] \colon \exists_{q(x) \in \mathbf{K}_{n}[x]} \;\; p(x) = q(x)^2 \} $$
Note that every solution to the phaseless interpolation problem for a subfield $\mathbf{K} \subseteq \mathbf{R}$ of the real numbers uniquely lifts to the solution of the interpolation problem of the square of the polynomial.

\section{Polynomial-time phaseless interpolation}

The problem of real phaseless polynomial interpolation of a polynomial of degree at most $n$ from $2n+1$ evaluation points was investigated in \cite{przybylek2020note}. Although their construction carries over to the rational case in a straightforward way, it does not easily generalise to the number of evaluation nodes less than $2n+1$ for the following two reasons. Firstly, their method is based on Lagrange interpolation of the square of a polynomial. Because the square of a polynomial of degree $n$ has degree $2n$, we need exactly $2n+1$ points for Lagrange interpolation procedure to work. Secondly, if we have less than $2n+1$ points, a polynomial might not be uniquely determined (up to a global phase) from the absolute-value evaluations.

Addressing the second issue, it was shown in \cite{przybylek2020note} that the adaptation is necessary if we restrict to $n+2$ rational information operations. Moreover, it is extremely difficult to hit rational $n+2$ nodes such that the reconstruction problem is ambiguous (in the real case, we have to hit a set of measure zero). Therefore, one may wonder if we can non-adaptively choose $n+3$ nodes, such that the phaseless interpolation is possible. Or more generally, what is the minimal number $k$ of non-adaptively chosen nodes $(x_i)_{i=0}^{k-1}$, such that the phaseless interpolation from $(x_i)_{i=0}^{k-1}$ gives a unique (up to a phase) polynomial. From \cite{przybylek2020note}, we know that $n+2 < k \leq 2n+1$. Contrary to our initial intuition, the next theorem shows that the minimal number of nodes is, in fact, $2n+1$.

\begin{theorem}[$2n$ points are not sufficient]\label{t:adaptation:necessary}
    Let $x_0, x_1, \cdots, x_{2n-1}$ be any $2n$ evaluation points. There exist polynomials $p, q \in \mathbf{Q}_n[x]$ such that $|p(x_i)| = |q(x_i)|$ for all $0 \leq i < 2n$, but $|p| \neq |q|$.
\end{theorem}
\begin{proof}
    Without loss of generality, let us assume that the nodes $x_i$ are pairwise distinct. Then polynomials: $p(x) = \prod_{i=0}^{n-1} (x-x_i) + \prod_{i=n}^{2n-1} (x-x_i)$ and: $q(x) = \prod_{i=0}^{n-1} (x-x_i) - \prod_{i=n}^{2n-1} (x-x_i)$ are of order $n$. Observe, that for $0 \leq i < n$ we have that $p(x_i) = -q(x_i)$, because the first terms in both of the polynomials are zero, and for $n \leq i < 2n$ we have that $p(x_i) = q(x_i)$, because the second terms in both of the polynomials are zero. Therefore, $|p(x_i)| = |q(x_i)|$ for every $0 \leq x < 2n$. On the other hand, clearly, $|p| \neq |q|$, which completes the proof.
\end{proof}

\begin{example}[Cubic polynomials on six nodes]
Consider the following evaluation points: $\overline{x} = [1, 2, 3, 4, 5, 6]$. Polynomials $p, q \in \mathbf{Q}_3[x]$ corresponding to these nodes are:
\begin{eqnarray*}
    p(x) &=& 2x^3 - 21x^2 + 85x - 126 \\
    q(x) &=& 9x^2 - 63x + 114 \\
\end{eqnarray*}

Figure \ref{fig:intersection} illustrates $|p|$ and $|q|$ and their intersections. 
\begin{figure}[H]
  \centering
  \includegraphics[width=0.8\textwidth]{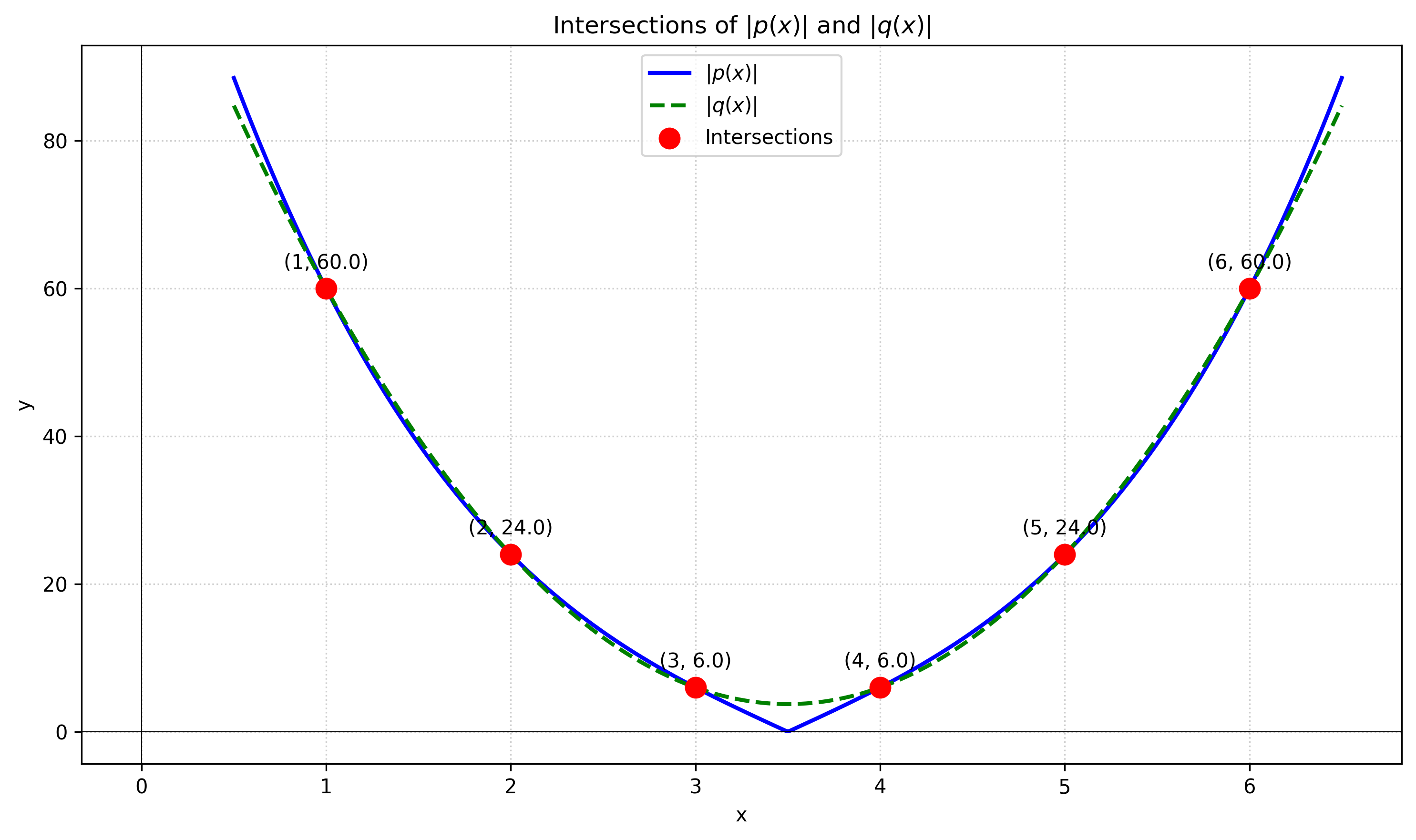}  
  \caption{Intersections of $|p(x)| = |2x^3 - 21x^2 + 85x - 126|$ and $|q(x)| = |9x^2 - 63x + 114|$.}
  \label{fig:intersection}
\end{figure}
\end{example}
The paper \cite{przybylek2020note} demonstrated only an exponential-time procedure for the selection of a single node from previously observed absolute values, such that the reconstruction problem is unambiguous. Therefore, for a polynomial-time reconstruction algorithm from less than $2n+1$ evaluation points, we need to improve the procedure of node-selection to work in polynomial-time. We shall address this issue below (Theorem~\ref{t:adaptation}), by showing that it is possible to select a single node adaptively in a polynomial time.

Addressing the first issue, we show that for a fixed $k \leq n$ there are at most polynomially-many solutions to the reconstruction problem from $2n-k+1$ absolute values and all of them can be found in a polynomial-time (Theorem~\ref{t:poly:time}). 

So, let $k \in \mathbb{N}$ be a fixed constant and consider a sequence $\tuple{x_0, y_0}, \tuple{x_1, y_1}, \dotsc, \tuple{x_{2n-k}, y_{2n-k}}$ for $n \geq k$, where $x_i$ are pairwise distinct and $y_i \geq 0$. Our task is to find all polynomials $q(x) \in \mathbf{Q}_n[x]$ such that $|q(x_i)| = y_i$ for all $i=0, \dots, 2n-k$.

We observe that the condition $|q(x_i)| = y_i$ is equivalent to $q(x_i)^2 = y_i^2$. Our first task is to find all polynomials $p \in \mathbf{Q}_{2n}[x]$ that interpolate given points $\tuple{x_0, y_0}, \tuple{x_1, y_1}, \dotsc, \tuple{x_{2n-k}, y_{2n-k}}$. By Lemma~\ref{l:affine} this set forms a $k$-dimensional affine space $\mathbf{Q}_{2n}(\overline{x}, \overline{y})[x]$, which is parametrised via $p(x, \cbold)$. 

Our task can be now rephrased as follows: find all vectors $\cbold = [c_0, c_1, \dots, c_{k-1}] \in \mathbf{Q}^k$ such that $p(x, \cbold) \in \mathbf{Q}_{n}^2[x]$. Observe, which will be crucial for the proof, that there cannot be to many such vectors $\cbold$.

\begin{lemma}[On the zero-dimensionality of solution sets]\label{l:zero:dim}
    If $n \geq k$ then there are only finitely many  $\cbold = [c_0, c_1, \dots, c_{k-1}] \in \mathbf{Q}^k$ such that $p(x, \cbold) \in \mathbf{Q}_{n}^2[x] $.
\end{lemma}
\begin{proof}
    If $n \geq k$, then we have at least $n+1$ interpolation points $\tuple{x_0, y_0}, \tuple{x_1, y_1}, \dotsc, \tuple{x_n, y_n}$. Note that $|p(x_i)| = y_i$ can be rephrased as: there exists $b_i \in\{-1, 1\}$ such that $p(x_i) = b_i y_i$. Moreover, for every such a vector $b \in \{-1, 1\}^{n+1}$ there exists a unique polynomial of degree at most $n$ interpolating $\tuple{x_0, b_0y_0}, \tuple{x_1, b_1y_1}, \dotsc, \tuple{x_n, b_ny_n}$. Therefore, polynomials satisfying the considered condition are defined by vectors $b \in \{-1, 1\}^{n+1}$. Since there are exactly $2^{n+1}$ such vectors, there are at most $2^{n+1}$ such polynomials, and so the number of solutions is bounded by $2^{n+1}$.
\end{proof}

If $p(x, \cbold) \in \mathbf{Q}_{n}^2[x] $ then it must be of the form $p(x, \cbold) = (a_0 + a_1 x + \cdots + a_n x^n)^2$ for some rational coefficients $a_i$. On the other hand, expanding $p(x, \cbold)$ we get $p(x, \cbold) = A_0(\cbold) + A_1(\cbold)x + \cdots + A_{2n}(\cbold)x^{2n}$, where the coefficients $A_i(\cbold)$, when treated as functions of $\cbold$, are \emph{affine} functions. Equating these two expansions, we get the following system of $2n+1$ polynomial equations in variables $c_0, c_1, \dotsc, c_{k-1}, a_0, a_1, \dotsc, a_n$:
\begin{align*}
A_i(\cbold) &= \left\{
  \begin{array}{ll}
    a_0^2 & \textnormal{if}\  i = 0  \\
    2a_0a_{i} + \sum_{j=1}^{i-1} a_{j}a_{i-j} & \textnormal{if}\  
    1 \leq i \leq 2n  \\
  \end{array} \right.
\end{align*}
where, for convenience, we set $a_{n+i} = 0$ for $1 \leq i \leq n$. Note, however, that the variables $a_1, a_2, \dotsc, a_n$ can be eliminated at the expense of moving to the \emph{rational} system of equations. We keep the equation $A_0(\cbold) = a_0^2$ and, assuming $a_0 \neq 0$, we use the equations for $1 \leq i \leq n$ to iteratively compute $a_i$ as follows:
$$a_i = \frac{A_{i}(\cbold) - \sum_{j=1}^{i-1} a_{j}a_{i-j}}{2a_0}$$
Then each $a_i$ becomes a rational function of $\cbold$ and $a_0$. Nonetheless, we can improve this slightly. Observe that without loss of generality we can assume that $0$ is among our interpolation points. For, let $\tuple{x_i, y_i}$ be any phaseless interpolation point with $y_i \neq 0$. There must be such a point, because a polynomial of degree $n$ cannot have more than $n$ roots. Thus, we may shift our affine space by $x_i$, i.e.~we consider the phaseless interpolation problem for $\tuple{x_0 - x_i, y_0}, \tuple{x_1 - x_i, y_1}, \dotsc, \tuple{x_{2n-k} - x_i, y_{2n-k}}$. Then the $i$-th point has its first coordinate zero and its second coordinate non-zero: we can reconstruct the shifted space of polynomials first, and then shift it back by $x_i$. Therefore, let us assume that $x_i = 0$ and $y_i \neq 0$. Then $p(0, \cbold) = L(0) = y_i^2 = a_0^2$. Therefore, $a_0 = \pm y_i$ corresponding to two global phases. By choosing $a_0 = y_i$ (equivalently $a_0 = -y_i$), our coefficients $a_i$ become polynomials in $\cbold$ of degree at most $i$.
Computed in the above way $a_i$ are substituted into remaining $n$ equations for $n < i \leq 2n$:
$$A_i(\cbold) = 2a_0a_{i} + \sum_{j=1}^{i-1} a_{j}a_{i-j}$$
Therefore, we are left with a system $S$ of $n$ polynomial equations of degree at most $2n$ in $k$ variables $c_0, c_1, \cdots, c_{k-1}$.
\begin{example}[Phaseless interpolation with $k=1$]\label{e:phaseless}
    Consider the problem of phaseless reconstruction of polynomial $q \in \mathbf{Q}_3[x]$ form the following evaluation nodes $x_i = i-3$ for $0 \leq i \leq 5$ with the corresponding absolute values: $y_0 = 8, y_1 = 2, y_2 = 2, y_3 = 2, y_4 = 4, y_5=2$ and $k=1$. Then:
    $$p(x, \cbold) = 4+(12c_0+4)x+(4c_0+8)x^2+(-15c_0+3)x^3+(-5c_0-2)x^4+(3c_0-1)x^5+c_0x^6$$
    Unfolding coefficients $a_i$ and fixing negative $a_0$ yields:
\begin{align*}
a_0 &= -2 \\
a_1 &= -3 c_0 - 1 \\
a_2 &= \frac{9 c_0^2 + 2 c_0 - 7}{4} \\
a_3 &= \frac{-27 c_0^3 - 15 c_0^2 + 49 c_0 + 1}{8} \\
a_4 &= \frac{405 c_0^4 + 324 c_0^3 - 650 c_0^2 - 156 c_0 + 77}{64} \\
a_5 &= \frac{-1701 c_0^5 - 1755 c_0^4 + 2826 c_0^3 + 1542 c_0^2 - 853 c_0 - 59}{128} \\
a_6 &= \frac{15309 c_0^6 + 19278 c_0^5 - 26325 c_0^4 - 22924 c_0^3 + 11707 c_0^2 + 3374 c_0 - 419}{512}
\end{align*}
Therefore, solutions to the reconstruction problem (with negative constant term) are tantamount to solutions of the following system of polynomial equations:
\begin{align*}
    \frac{405 c_0^4 + 324 c_0^3 - 650 c_0^2 - 156 c_0 + 77}{64} &=&0\\
\frac{-1701 c_0^5 - 1755 c_0^4 + 2826 c_0^3 + 1542 c_0^2 - 853 c_0 - 59}{128} &=&0\\
\frac{15309 c_0^6 + 19278 c_0^5 - 26325 c_0^4 - 22924 c_0^3 + 11707 c_0^2 + 3374 c_0 - 419}{512} &=&0
\end{align*}
The polynomials together with their roots are presented on Figure~\ref{fig:roots}. Notice that there is a single common root at $c_0 =1$, but at $c_0 \approx -1.6$, the roots are distinct, one may compute the root of the first polynomial around this point to be (yes, it is a real root):
$$c_0 = -\frac{3}{5} - \frac{1}{3}\sqrt[3]{\frac{632}{375} + \frac{8 i \sqrt{9151959}}{6075}} - \frac{1792}{2025 \sqrt[3]{\frac{632}{375} + \frac{8 i \sqrt{9151959}}{6075}}}$$
\begin{figure}[H]
  \centering
  \includegraphics[width=0.8\textwidth]{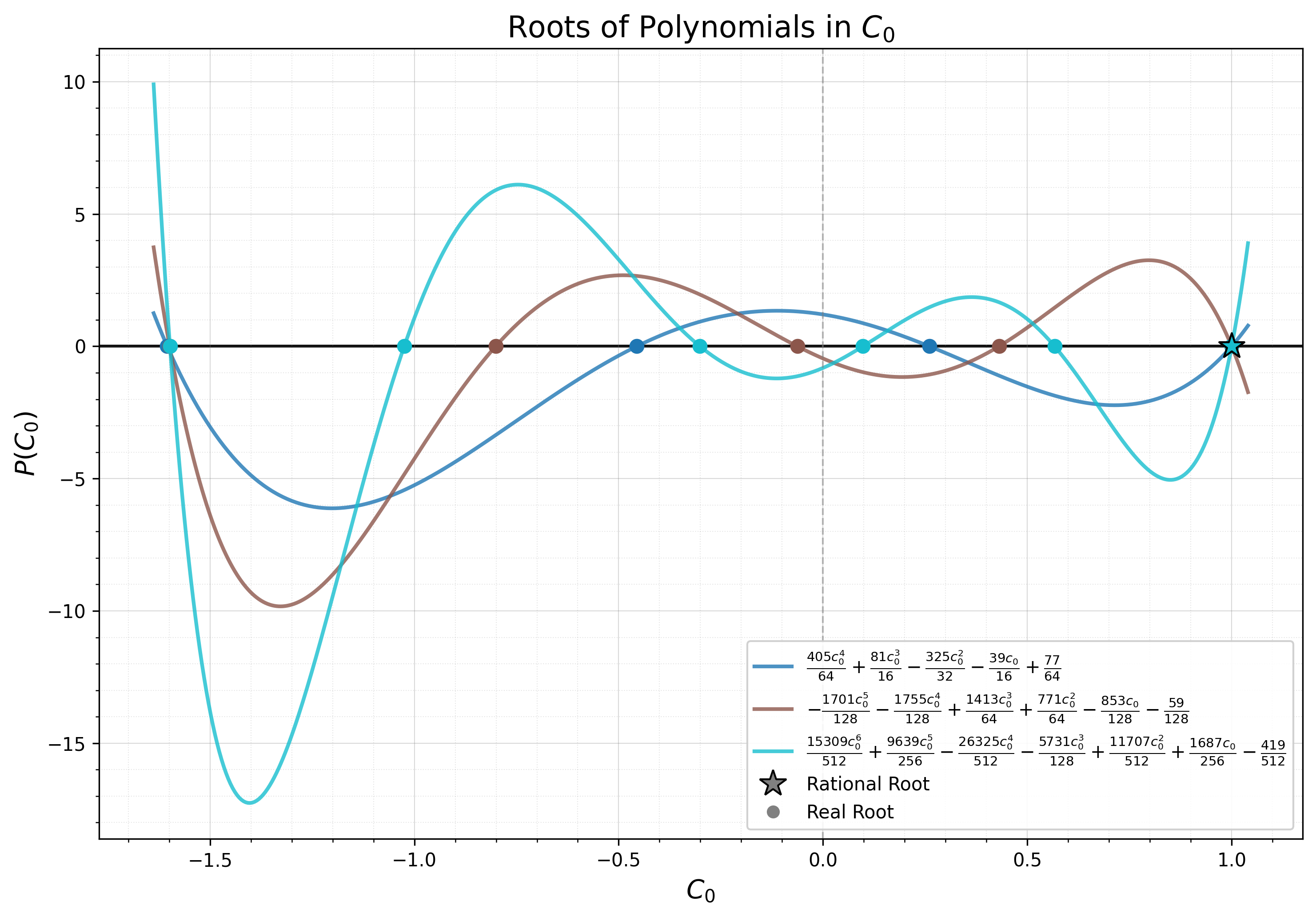}  
  \caption{Rational and all real roots for polynomials $a_4, a_5, a_6$ from Example~\ref{e:phaseless}.}
  \label{fig:roots}
\end{figure}
Therefore, we have a single solution (up to a phase) at $c_0 = 1$:
\begin{align*}
a_0 &= -2 \\
a_1 &= -4 \\
a_2 &= 1 \\
a_3 &= 1 \\
\end{align*}
That is: $q(x) = -2 -4x + x^2 + x^3$.
\end{example}

Let us fix the order $c_0 < c_1 < \cdots < c_{k-1}$ on the variables and let $\mathcal{B}$ be the Gröbner basis for $S$ for the lexicographical ordering. We need the following standard Lemma \cite{becker1993grobner}, \cite{cox1997ideals}, \cite{cox1998using}.

\begin{lemma}[Univariate polynomial in Gröbner basis]\label{l:triangular}
    Let $\mathcal{B}$ be a zero-dimensional Gröbner basis for the lexicographical ordering and $\overline{x} = x_1 < x_2 < \cdots < x_n$. Then $\mathcal{B}$ contains a univariate polynomial in $x_1$.
\end{lemma}
\begin{proof}
Let $I = \tuple{\mathcal{B}}$ be the ideal generated by $\mathcal{B}$ and denote by $V = \Q[\overline{x}]/I$ the quotient of the ring of polynomials in variables $x_1, x_2, \dotsc, x_n$ by ideal $I$. Observe that $V$ treated as a vector space over $\Q$ is finite dimensional. Two polynomials $p, q \in \Q[\overline{x}]$ are equal in $V$ if $p - q \in I$, that is, if $p(s) = q(s)$ for ever $s$ such that $I(s)=0$. Because $\mathcal{B}$ is zero-dimensional, there are only finitely many such $s$, say $s_1, s_2, \dotsc, s_k$. We claim that polynomials $\lambda_i(\overline{x}) = \prod_{j \neq i} \frac{\overline{x} - s_j}{s_i - s_j}$ form a basis of $V$. They are linearly independent since for a given point $s_i$ only $p_i(s_i) \neq 0$. Now, consider any $p \in \Q[\overline{x}]$. Then, by definition of $V$, polynomial $p$ is equal to $\sum_i p(s_i)\lambda_i$ in $V$. Therefore $\lambda_i$ span $V$, and so $V$ is $k$-dimensional.

Consider the following univariate polynomial in variable $x_1$: $1 + x_1 + x_1^2 + \dotsc + x_1^k$, which is a sum of $k+1$ vectors. Because $V$ is $k$-dimensional, there are coefficients $[a_0, a_1, \dotsc, a_k] \neq 0$ such that $p(x_1) = a_0 + a_1x_1 + a_2x_1^2 + \dotsc + a_kx_1^k = 0$ in $V$. Therefore $p \in I$, hence $I \cap \Q[x_1] \neq \emptyset$. By the Elimination Theorem for Gröbner basis, $\mathcal{B} \cap \Q[x_1]$ is the Gröbner basis for $I \cap \Q[x_1]$, therefore $\mathcal{B} \cap \Q[x_1] \neq \emptyset$.
\end{proof}

By Lemma~\ref{l:zero:dim} the basis $\mathcal{B}$ is zero-dimensional, and by Lemma~\ref{l:triangular} it contains an univariate polynomial $u(c_0)$. Because the degree of $u(c_0)$ is bounded by $2n$, it has at most $2n$ solutions $c_0^*$. 

Now we repeat the above process for every variable $c_i$ for $i > 0$. For each solution $c_{i-1}^*$ from the previous solution set (i.e.~starting from $c_0^*$) we substitute $c_{i-1}^*$ for its corresponding variable $c_i$ in $\mathcal{B}$, compute the new basis $\mathcal{B}[c_{i-1}^*/c_{i-1}]$ for polynomial equations of one less variables and find the roots of the univatiate polynomial in the variable $c_i$.

Each sequence $c_0^*, c_1^*, \dotsc, c_{k-1}^*$ found in the above process is tantamount to an interpolating polynomial. Notice that there are at most $(2n)^{k}$ such sequences, thus for a fixed parameter $k$ there are only polynomially-many solutions to the phaseless interpolation problem. This observation yields the following Lemma.

\begin{lemma}[On the number of solutions to the phaseless interpolation problem]\label{l:poly:sol}
    Let $k \in \mathbb{N}$ be a fixed constant. For $n \geq k$ and a sequence $\tuple{x_0, y_0}, \tuple{x_1, y_1}, \dotsc, \tuple{x_{2n-k}, y_{2n-k}}$, where $x_i$ are pairwise distinct and $y_i \geq 0$, the number of polynomials $q(x) \in \mathbf{Q}_n[x]$ such that $|q(x_i)| = y_i$ is $O(n^k)$.
\end{lemma}

\begin{example}[Phaseless interpolation with $k=2$]\label{e:phaseless2d}
Consider the polynomial from Example~\ref{e:phaseless} with evaluation node $x_5 = 3$ removed. Then for $k=2$ we get:
$$p(x, \cbold) = 4 + (4c_{0} + 8)x + (4c_{1} + 8)x^{2} + (-5c_{0} - 2)x^{3} + (-5c_{1} - 2)x^{4} + c_{0}x^{5} + c_{1}x^{6}$$
Unfolding coefficients $a_i$ and fixing negative $a_0$ gives us:
\begin{align*}
a_0 &= -2 \\[4pt]
a_1 &= -c_0 - 2 \\[4pt]
a_2 &= \frac{c_0^2}{4} + c_0 - c_1 - 1 \\[6pt]
a_3 &= \frac{5c_0}{4} 
      -\frac{(c_0 + 2)(c_0^2 + 4c_0 - 4c_1 - 4)}{8} 
      + \frac{1}{2} \\[10pt]
a_4 &= \frac{5c_0^4}{64} 
      + \frac{5c_0^3}{8} 
      - \frac{3c_0^2 c_1}{8} 
      + \frac{c_0^2}{2} 
      - \frac{3c_0 c_1}{2} 
      - 2c_0 \\ 
    &\quad + \frac{c_1^2}{4} 
      + \frac{3c_1}{4} 
      - \frac{3}{4} \\[10pt]
a_5 &= -\frac{7c_0^5}{128} 
      - \frac{35c_0^4}{64} 
      + \frac{5c_0^3 c_1}{16} 
      - \frac{35c_0^3}{32} 
      + \frac{15c_0^2 c_1}{8} \\ 
    &\quad + \frac{23c_0^2}{16} 
      - \frac{3c_0 c_1^2}{8} 
      + c_0 c_1 
      + \frac{5c_0}{2} 
      - \frac{3c_1^2}{4} 
      - 2c_1 \\[10pt]
a_6 &= \frac{21c_0^6}{512} 
      + \frac{63c_0^5}{128} 
      - \frac{35c_0^4 c_1}{128} 
      + \frac{195c_0^4}{128} \\ 
    &\quad - \frac{35c_0^3 c_1}{16} 
      - \frac{5c_0^3}{16} 
      + \frac{15c_0^2 c_1^2}{32} 
      - \frac{105c_0^2 c_1}{32} 
      - \frac{285c_0^2}{64} \\ 
    &\quad + \frac{15c_0 c_1^2}{8} 
      + \frac{23c_0 c_1}{8} 
      - \frac{21c_0}{16} 
      - \frac{c_1^3}{8} \\ 
    &\quad + \frac{c_1^2}{2} 
      + \frac{5c_1}{2} 
      + \frac{15}{16}
\end{align*}
Figure~\ref{fig:roots2d} shows roots of polynomials $a_4(c_0, c_1) = 0$, $a_5(c_0, c_1) = 0$ and $a_6(c_0, c_1) = 0$. There is a single common root at $(c_0, c_1) = (2, 1)$, which yields:
\begin{align*}
a_0 &= -2 \\
a_1 &= -4 \\
a_2 &= 1 \\
a_3 &= 1 \\
\end{align*}
That is: $q(x) = -2 -4x + x^2 + x^3$ again. To prove that $(c_0, c_1) = (2, 1)$, we can compute the Gröbner basis for  $a_4(c_0, c_1) = 0$, $a_5(c_0, c_1) = 0$ and $a_6(c_0, c_1) = 0$ and observe that it consists of two polynomial equations:
\begin{align*}
c_0 - 2 &= 0 \\
c_1 - 1 &= 0 \\
\end{align*}

\begin{figure}[H]
  \centering
  \includegraphics[width=0.8\textwidth]{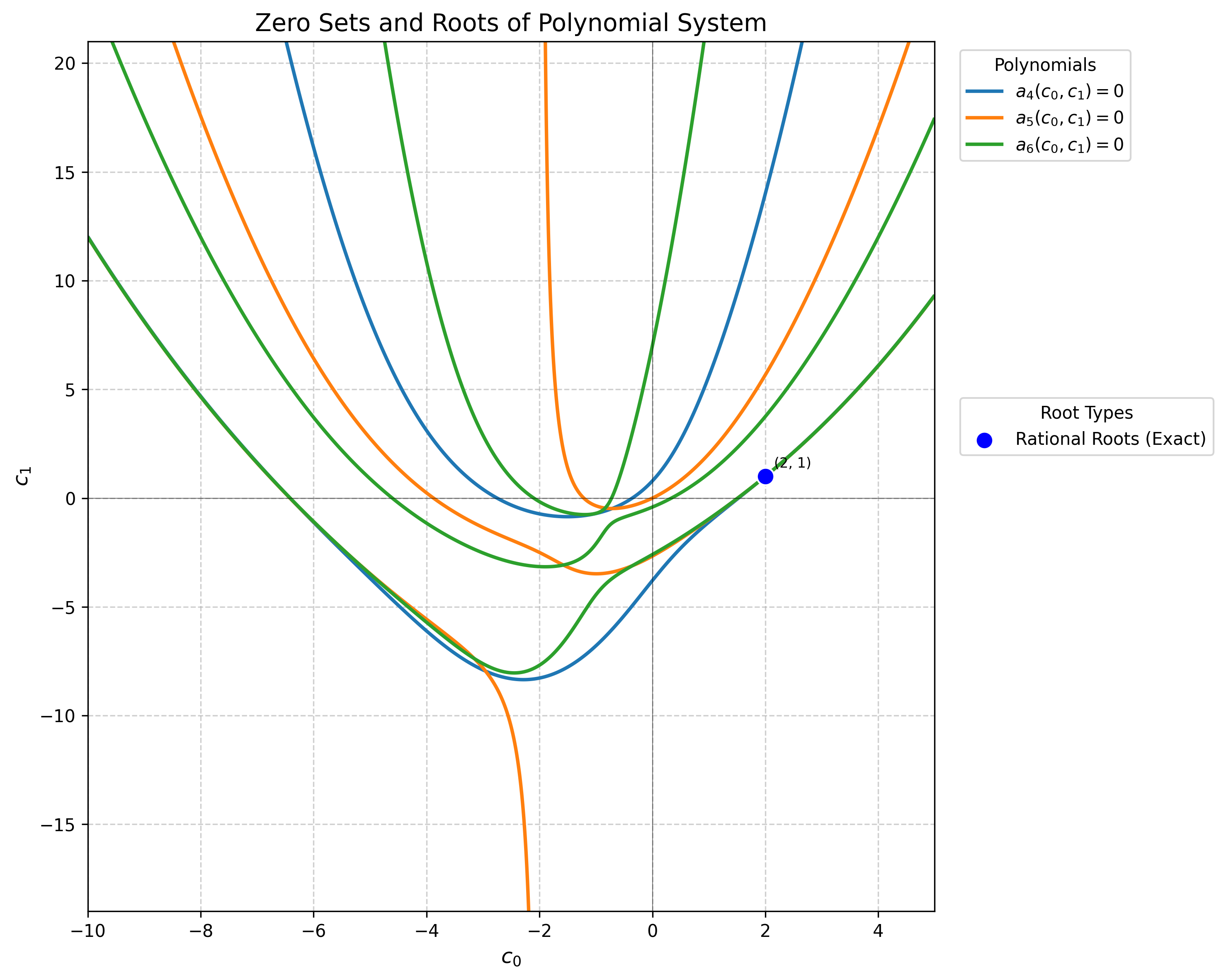}  
  \caption{Real roots of polynomials $a_4, a_5, a_6$ from Example~\ref{e:phaseless2d}. There is a single common root at $(c_0, c_1) = (2, 1)$.}
  \label{fig:roots2d}
\end{figure}

\end{example}

\begin{example}[Phaseless interpolation with $k=3$]\label{e:phaseless3d}
In case $k=3$ we need $n>3$, otherwise there will be a unique interpolating polynomial for every choice of signs in the evaluation values. Let us consider the following polynomial $q \in \mathbf{Q}_4[x]$:
$$q(x) = x^4+x^3-10x^2-4^x+9$$
with evaluation nodes $x_i = i-2$ for $0 \leq i \leq 5$ and $k=3$. Then:
\begin{align*}
p(x, \cbold) &=  81 \\
    &\quad + (-12c_{0} - 60) x \\
    &\quad + (4c_{0} - 12c_{1} - 108) x^{2} \\
    &\quad + (15c_{0} + 4c_{1} - 12c_{2} + 75) x^{3} \\
    &\quad + (-5c_{0} + 15c_{1} + 4c_{2} + 36) x^{4} \\
    &\quad + (-3c_{0} - 5c_{1} + 15c_{2} - 15) x^{5} \\
    &\quad + (c_{0} - 3c_{1} - 5c_{2}) x^{6} \\
    &\quad + (c_{1} - 3c_{2}) x^{7} \\
    &\quad + c_{2} x^{8}
\end{align*}
The solution to the polynomial system $a_5 = 0, a_6 = 0, a_7 = 0, a_8 = 0$ with $a_0=9$ is presented on Figure~\ref{fig:roots3d}. The Gröbner basis of the system consists of the following polynomials:
\begin{align*}
\dfrac{116}{7371} c_0^3 + \dfrac{788}{2457} c_0^2 + \dfrac{92}{63} c_0 + c_2 - \dfrac{2945}{1053} \\
-\dfrac{64}{2457} c_0^3 - \dfrac{1571}{2457} c_0^2 - \dfrac{725}{189} c_0 + c_1 - \dfrac{175}{351} \\
c_0^4 + \dfrac{67}{2} c_0^3 + \dfrac{5649}{16} c_0^2 + \dfrac{8257}{8} c_0 - \dfrac{22715}{16}.
\end{align*}
The roots of the polynomials that form the Gröbner basis are presented on Figure~\ref{fig:roots3dGB}. Notice, there are four unique solutions (up to a phase), which are rational. The solutions yield the following polynomials:
\begin{align*}
(-\frac{59}{4}, -\frac{9}{16}, \frac{81}{16}) & \mapsto 2.25x^4-3.5x^3-11.25x^2+6.5x+9\\
(-11, 1, 1) & \mapsto x^4-x^3-10x^2+4x+9\\
(-\frac{35}{4}, -\frac{25}{16}, \frac{25}{16}) & \mapsto 1.25x^4-2.5x^3-7.25x^2+2.5x+9\\
(1, 5, 1) & \mapsto x^4+x^3-10x^2-4x+9
\end{align*}

\begin{figure}[H]
  \centering
  \includegraphics[width=0.8\textwidth]{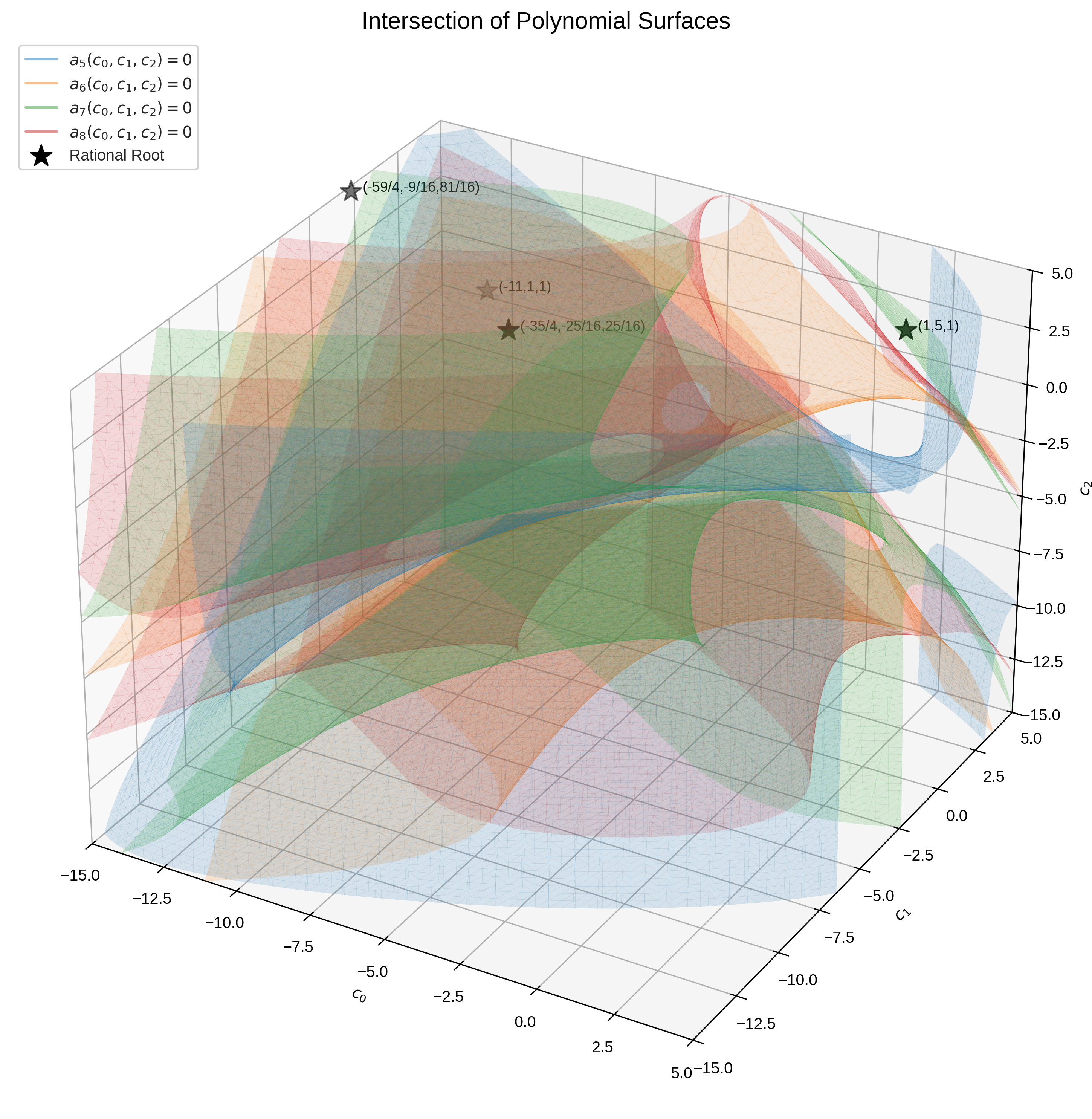}  
  \caption{Real roots of polynomials $a_5, a_6, a_7, a_8$ from Example~\ref{e:phaseless3d}. There are four common roots at $(-\frac{59}{4}, -\frac{9}{16}, \frac{81}{16})$, $(-11, 1, 1)$, $(-\frac{35}{4}, -\frac{25}{16}, \frac{25}{16})$ and $(1, 5, 1)$.}
  \label{fig:roots3d}
\end{figure}

Figure~\ref{fig:2dslices} shows 2D slices of Figure~\ref{fig:roots3d} at four different values of $c_0$ corresponding to four different common roots of polynomials $a_5, a_6, a_7$ and $a_8$. Note, however, that if we change the evaluation point from $x_5=3$ to $x_5=4$, then the system will have a unique solution.    
\end{example}

\begin{theorem}[Polynomial-time reconstruction]\label{t:poly:time}
Let $k \in \mathbb{N}$ be a fixed constant. The following problem has polynomial-time complexity. Given $n \geq k$ and a sequence 
$$\tuple{x_0, y_0}, \tuple{x_1, y_1}, \dotsc, \tuple{x_{2n-k}, y_{2n-k}},$$ 
where $x_i$ are pairwise distinct and $y_i \geq 0$, find all polynomials $q(x) \in \mathbf{Q}_n[x]$ such that $|q(x_i)| = y_i$ for all $i=0, \dots, 2n-k$.
\end{theorem}
\begin{proof}
The process of computing the coefficients $A_i(\cbold)$ of $p(x, \cbold)$ consists of the computation of the coefficients of the interpolating polynomial $L(x)$ and the coefficients of $\prod_{i=0}^{n-k} (x-x_i)$, which are of bit-size polynomial in the bit-size of $\tuple{x_0, y_0}, \tuple{x_1, y_1}, \dotsc, \tuple{x_{2n-k}, y_{2n-k}}$. Therefore, it is polynomial-time. Since each $a_i$ has degree bounded by $2i-1$, the process of computing coefficients of $a_i$ is polynomial in bit-size of coefficients of $A_i(\cbold)$. Therefore, it is in polynomial-time. The number of basic operations needed to compute the Gröbner basis is bounded by $(2n)^{2^{k}}$, hence, for a fixed parameter $k$, its time-complexity is polynomial. By the classical result of A.K.~Lenstra, H.W.~Lenstra and L.~Lov{\'a}sz \cite{lenstra1982factoring}, finding rational roots of a rational polynomial is polynomial-time of the bit-size of the coefficients. Therefore, the whole process is polynomial in the bit-size of $\tuple{x_0, y_0}, \tuple{x_1, y_1}, \dotsc, \tuple{x_{2n-k}, y_{2n-k}}$.

\begin{figure}[H]
  \centering
  \includegraphics[width=0.8\textwidth]{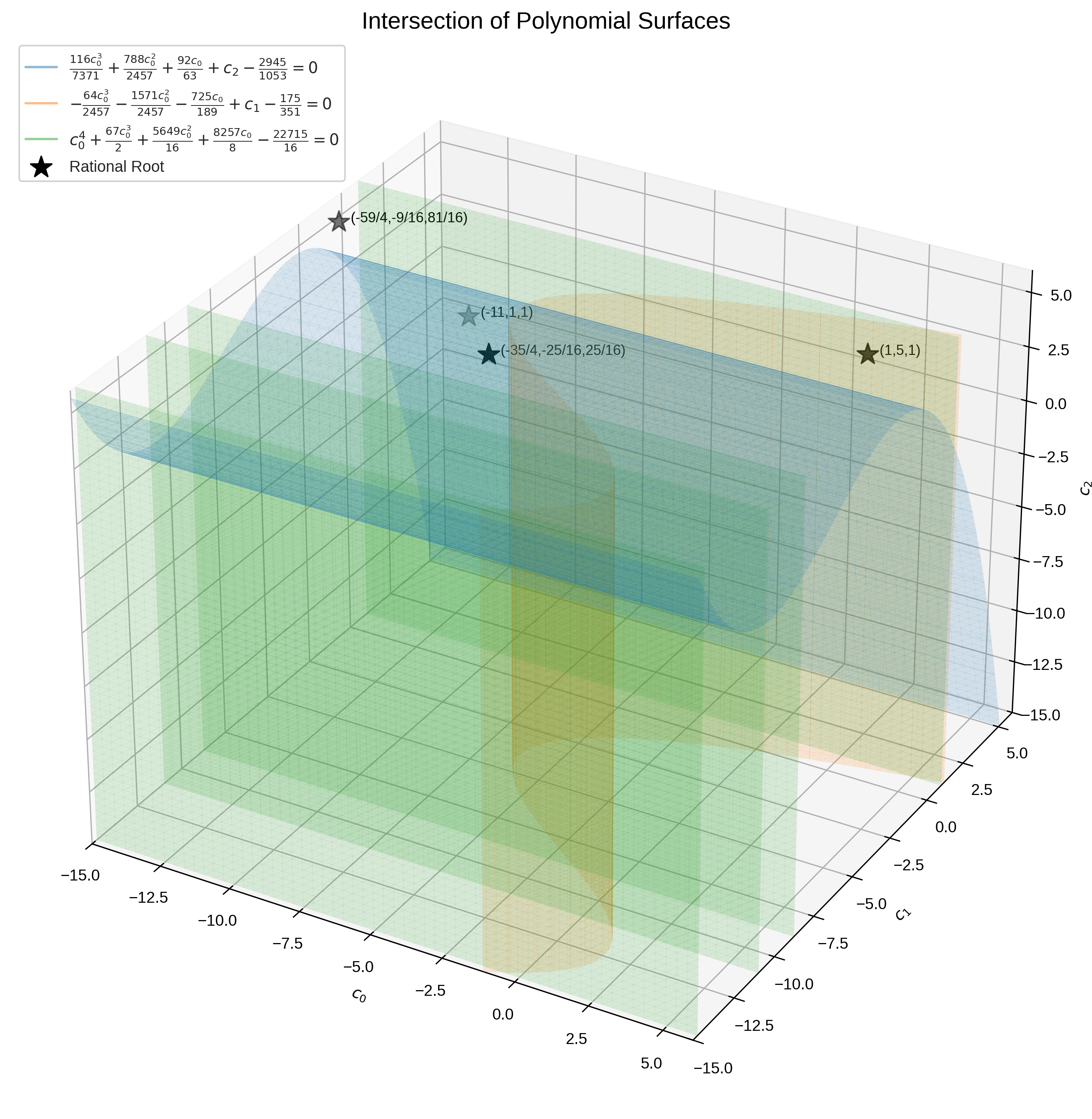}  
  \caption{Real roots of Gröbner basis for $a_5, a_6, a_7, a_8$ from Example~\ref{e:phaseless3d}. There are four common roots at $(-\frac{59}{4}, -\frac{9}{16}, \frac{81}{16})$, $(-11, 1, 1)$, $(-\frac{35}{4}, -\frac{25}{16}, \frac{25}{16})$ and $(1, 5, 1)$.}
  \label{fig:roots3dGB}
\end{figure}
\begin{figure}[htbp]
    \centering
    
    \begin{subfigure}[b]{0.45\textwidth}
        \centering
        \includegraphics[width=\textwidth]{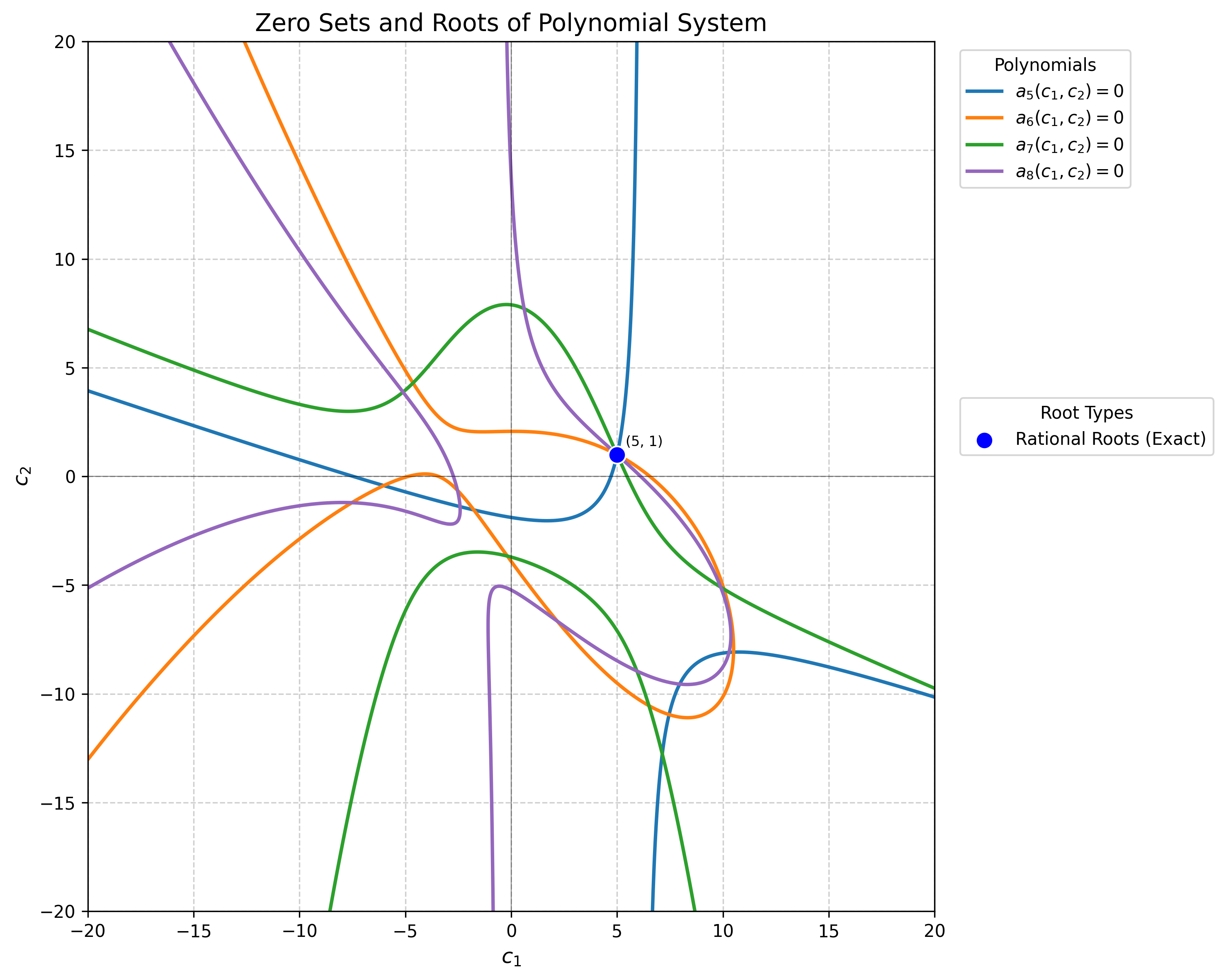}
        \caption{Polynomials with $c_0=1$}
        \label{fig:top-left}
    \end{subfigure}
    \hfill 
    \begin{subfigure}[b]{0.45\textwidth}
        \centering
        \includegraphics[width=\textwidth]{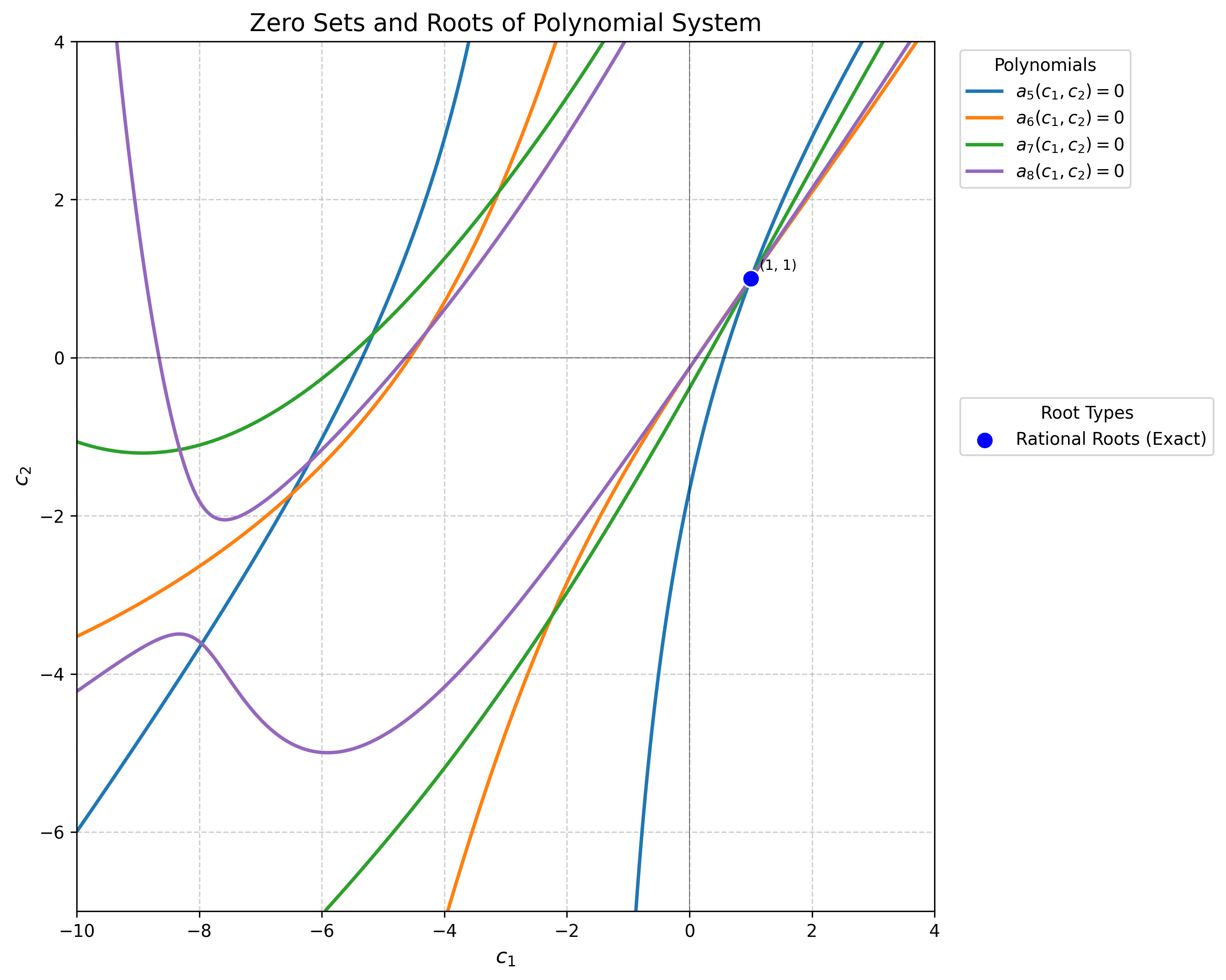}
        \caption{Polynomials with $c_0=-11$}
        \label{fig:top-right}
    \end{subfigure}
    
    \vspace{0.5cm} 
    
    \begin{subfigure}[b]{0.45\textwidth}
        \centering
        \includegraphics[width=\textwidth]{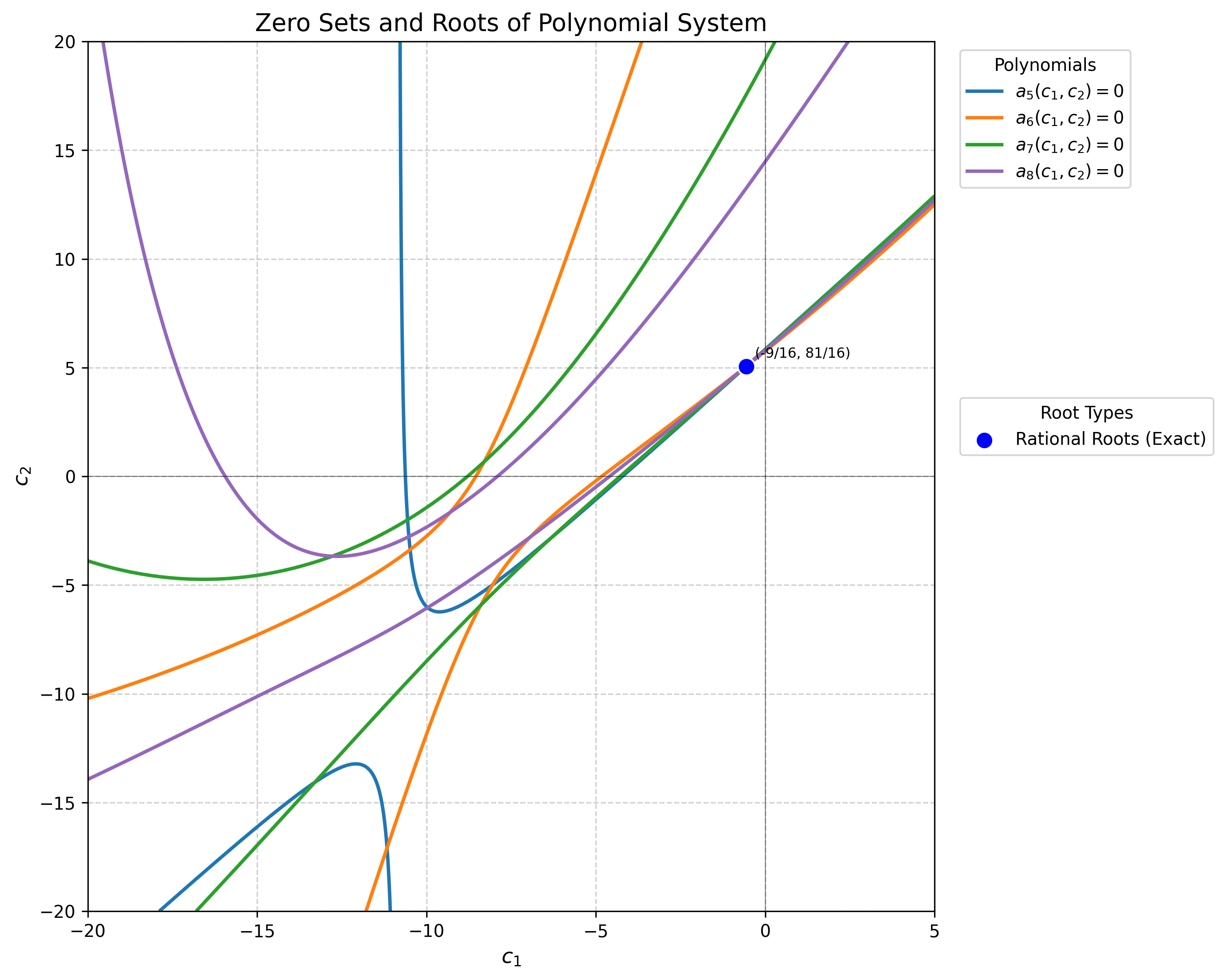}
        \caption{Polynomials with $c_0=-\frac{59}{4}$}
        \label{fig:bottom-left}
    \end{subfigure}
    \hfill
    \begin{subfigure}[b]{0.45\textwidth}
        \centering
        \includegraphics[width=\textwidth]{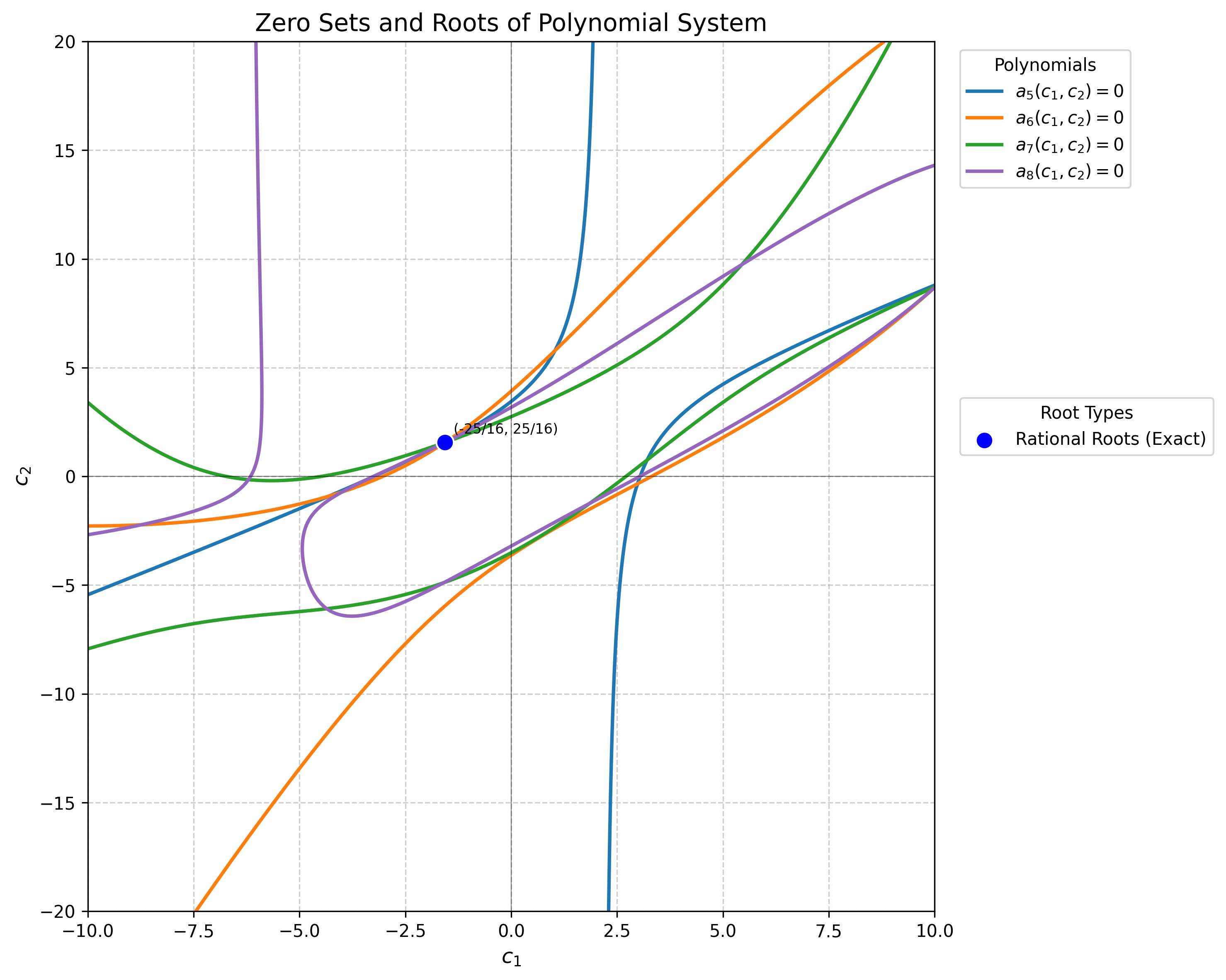}
        \caption{Polynomials with $c_0=-\frac{35}{4}$}
        \label{fig:bottom-right}
    \end{subfigure}
    
    \caption{2D slices of real roots of polynomials $a_5, a_6, a_7, a_8$ from Example~\ref{e:phaseless3d}.}
    \label{fig:2dslices}
\end{figure}
\end{proof}

\begin{remark}
 Although Lemma~\ref{l:triangular} is sufficient to prove that the phaseless interpolation problem form $2n-k$ points is in polynomial-time, it can be strengthened a bit to show that after substitution for the first variable, the Gröbner basis does not need to be recomputed. 
\end{remark}

An educational implementation of the algorithm is presented on Listing 1 in the Appendix.

\begin{theorem}[Phaseless Interpolation as an IBC problem] \label{t:adaptreal}
For a fixed parameter $k$, the Phaseless Interpolation Problem for polynomials of degree at most $n$ has a solution $(N,\phi)$ such that $N$ uses $n-k$ basic information operations with a single adaptation (i.e., $N$ selects a single point adaptively)  
whose computational complexity is bounded by $O(n^{\alpha_k})$ for some constant $\alpha_k$ that depends only on $k$ but does not depend on $n$.
\end{theorem}
\begin{proof}
    By Theorem~\ref{t:poly:time} and Theorem~\ref{t:adaptation} below.
\end{proof}

\begin{remark}
    By a more careful analysis of the problem, one may show that in Theorem~\ref{t:adaptreal} one may set $\alpha_k = k$. Nonetheless, in our analysis we are primarily interested in the gap between polynomial-time and NP-complete computational classes.
\end{remark}

\begin{theorem}[Single adaptation in polynomial-time]\label{t:adaptation}
    Let $x_0, x_1, \dots, x_{n}$ be $n+1$ distinct evaluation points and let $y_0, y_1, \dots, y_{n}$ be the corresponding absolute values of the evaluations of a polynomial $p$ of degree at most $n$. There exists a point $x_{n+1}$ such that $y_{n+1} = |p(x_{n+1})|$ uniquely determines $|p|$. 
    
    Moreover, if the polynomial and evaluation points are rational, the point $x_{n+1}$ can be computed from 
    $\{\tuple{x_i, y_i}\}_{i=0}^n$ in time polynomial in the bit-size of the input.
\end{theorem}

\begin{proof}
    The existence of such a point $x_{n+1}$ was established in \cite{przybylek2020note}, but the constructive method provided therein required exponential time. We present a construction that operates in polynomial time.

    Let $S$ be the set of all polynomials interpolating points $\{\tuple{x_i, b_i y_i}\}_{i=0}^n$ for every possible sign vector $b \in \{-1, 1\}^{n+1}$. The absolute value $|p|$ is uniquely determined if and only if no two distinct polynomials in $S$ intersect at $x_{n+1}$.
    
    Consider any two distinct sign vectors $b, b' \in \{-1, 1\}^{n+1}$. The intersection of the corresponding polynomials $L_b$ and $L_{b'}$ occurs as a root of their difference:
    $$
        D_{b,b'}(x) = L_b(x) - L_{b'}(x) = \sum_{j=0}^n (b_j - b'_j) y_j \ell_j(x),
    $$
    where $\ell_j(x)$ are the Lagrange basis polynomials. We seek an integer $x_{n+1}$ that is not a root of $D_{b,b'}(x)$ for any pair $b, b'$.

    Since the inputs $x_i, y_i$ are rational, we can clean denominators to work with integers. Let $M$ be a common multiple of the denominators of all coefficients in the expansions of $y_j \ell_j(x)$. Define the integer polynomial $P_{b,b'}(x) = M \cdot D_{b,b'}(x)$. 
    
    We derive an easily computable bound $B$ on the magnitude of integer roots of $P_{b,b'}(x)$ that is independent of the specific signs $b, b'$. Let $P_{b,b'}(x) = \sum_{k=0}^n A_k(b,b') x^k$. By the Rational Root Theorem, any non-zero integer root of $P_{b,b'}(x)$ must divide the lowest-degree non-zero coefficient $A_s(b,b')$. Therefore, any such root is bounded by $|A_s(b,b')| \le \max_k |A_k(b,b')|$. We can bound the coefficients uniformly using the triangle inequality. Since $|b_j - b'_j| \le 2$, we have:
    $$
        |A_k(b,b')| \le \sum_{j=0}^n |b_j - b'_j| \cdot |\text{coeff of } x^k \text{ in } M y_j \ell_j(x)| \le 2 \sum_{j=0}^n |c_{j,k}|
    $$
    where $c_{j,k}$ is the coefficient of $x^k$ in the polynomial $M y_j \ell_j(x)$. Let $B = \max_k \left( 2 \sum_{j=0}^n |c_{j,k}| \right)$. Such a bound $B$ depends only on the inputs $x_i, y_i$.
    
    Choosing an integer $x_{n+1} = B + 1$ ensures that $x_{n+1}$ is strictly larger than every possible integer root of any difference polynomial. Thus, no two polynomials in $S$ intersect at $x_{n+1}$. Since $B$ is computed via basic arithmetic operations on the rational inputs, its bit-size is polynomial in the input size.
\end{proof}

\section{The case of $(1+c)n$ points for $0 \leq c < 1$}

In this section we show that the phaseless polynomial retrieval from $(1+c)n$ evaluation points is hard for any fixed $0 \leq c < 1$. In fact, it remains hard in case $k = \lfloor cn \rfloor$ evaluations are exact (i.e.~with phase). However, one must be very careful when formulating the claim. The theorem stated below says that the retrieval problem is NP-hard. Nonetheless, there are two ways to state the theorem--a weak one: ``for every $n, k$ we can choose evaluation nodes, such that the problem is hard'', and a strong one: ``for every $n, k$ no matter how we choose evaluation nodes, the problem is hard''. In the sequel we prove the stronger version. Note that in the strong version we have to be prudent with specifying what is the input to our problem. We clearly cannot take the evaluation nodes as inputs, because we want to restrict to one particular choice for every $n$. But then, if the bit-size of the $n$th evaluation node is super-polynomial in $n$, then we cannot perform any arithmetic in polynomial time. One way to make sense of the above is as follows: we prove the theorem for any infinite sequence of evaluation nodes $x_0, x_1, x_2, \dots$ such that the bit-size of $x_i$ is polynomial in $i$. 

\begin{theorem}[Phaseless polynomial retrieval over $\tuple{\overline{r}, \overline{v}}$ is NPC]
Let $r_0, r_1, r_2, \dots$ and $v_0, v_1, v_2, \dots$ be two infinite sequences of nodes, where $r_i$ are pairwise distinct, and the bit-size of $r_i$ and $v_i$ is polynomial in $i$.  
The problem of identifying a polynomial $p \in \mathbf{Q}_{n}\left((r_i)_{i=0}^{k-1}, (v_i)_{i=0}^{k-1}\right)[x]$ up to its phase from non-adaptive $n-k+2$ phaseless evaluations at points $x_0 = r_k, x_1 = r_{k+1}, \dots, x_{n-k+1} = r_{n+1}$ is NP-complete for $k = \lfloor cn \rfloor$, where $0 \leq c < 1$.
\end{theorem}

\begin{proof}
Recall from Lemma~\ref{l:affine} the parametrisation of the affine space of interpolating polynomials:
$$\cbold \in \Q^{n-k+2} \mapsto  p(x, \cbold) = L(x) + \left( \sum_{j=0}^{n-k} c_j x^j \right) \prod_{j=0}^{k-1} (x-r_j)$$
Let distinct points $x_0, x_1, \dots, x_{n-k+1} \in \mathbf{Q}$ be given as above, together with the absolute values $y_0, y_1, \dots, y_{n-k+1} \in \mathbb{Q}$. Let $p \in p(x, \cbold)$, then for some $b_i \in\{-1, 1\}$ we have that:
$$b_i y_i = b_i|p(x_i)| = p(x_i) =  L(x_i) + \left( \sum_{j=0}^{n-k} c_j x_i^j \right) \prod_{j=0}^{k-1} (x_i-r_j)$$
Therefore:
$$\sum_{j=0}^{n-k} c_j x_i^j = \frac{b_iy_i - L(x_i)}{ \prod_{j=0}^{k-1}(x_i - r_j)}$$
If we assume that $b_i$ are fixed, then the above is the interpolation problem for a polynomial of degree at most $n-k$ from $n-k+2$ interpolation points, i.e.: $c_j$ are the solutions to the over-constrained system of equations $V [c_0, c_1, \dotsc, c_{n-k}] = v$, where:
$$
V = \begin{pmatrix}
1 & x_0 & x_0^2 & \cdots & x_0^{n-k} \\
1 & x_1 & x_1^2 & \cdots & x_1^{n-k} \\
1 & x_2 & x_2^2 & \cdots & x_2^{n-k} \\
\vdots & \vdots & \vdots & \ddots & \vdots \\
1 & x_{n-k+1} & x_{n-k+1}^2 & \cdots & x_{n-k+1}^{n-k}
\end{pmatrix}$$
$$v =
\begin{pmatrix}
 \frac{b_0y_0 - L(x_0)}{ \prod_{j=0}^{k-1}(x_0 - r_j)} \\
 \frac{b_1y_1 - L(x_1)}{ \prod_{j=0}^{k-1}(x_1 - r_j)} \\
 \frac{b_2y_2 - L(x_2)}{ \prod_{j=0}^{k-1}(x_2 - r_j)} \\
\vdots \\
 \frac{b_{n-k+1}y_{n-k+1} - L(x_{n-k+1})}{ \prod_{j=0}^{k-1}(x_{n-k+1} - r_j)} \\
\end{pmatrix}
$$
System $V \cbold = v$ has solutions if and only if $v$ is in the image $\mathit{Im}(V)$ of $V$. If we denote by $\mathit{Im}(V)^\bot$ the vector space orthogonal to $\mathit{Im}(V)$, then $v \in \mathit{Im}(V)$ if and only if $v \bot \mathit{Im}(V)^\bot$. Moreover, since $V$ is of full column rank, the space $\mathit{Im}(V)^\bot$ is 1-dimensional and can be represented by a single basis vector $w \in \mathit{Im}(V)^\bot$. Then, $v \in \mathit{Im}(V)$ if and only if $\langle v, w \rangle = 0$.
Unfolding the expression, we get:
$$\sum_{i=0}^{n-k+1} \frac{b_iy_i - L(x_i)}{ \prod_{j=0}^{k-1}(x_i - r_j)} w_i = 0$$
or 
$$\sum_{i=0}^{n-k+1} b_iy_i\alpha_i - \beta_i = 0$$
where: $\alpha_i = \frac{w_i}{\prod_{j=0}^{k-1}(x_i - r_j)}$ and $\beta_i = \frac{L(x_i)w_i}{\prod_{j=0}^{k-1}(x_i - r_j)}$ are non-zero and do not depend on values $y_i$. Furthermore, if $S = \sum_{i=0}^{n-k+1} \beta_i$, then the above equation can be rewritten as:
$$\sum_{i=0}^{n-k+1} b_i y_i \frac{\alpha_i}{S} = 1$$
\begin{remark}[On non-zero coordinates]
    We have to ensure that all $w_i \neq 0$. But this is, indeed, the case and we can even find an explicit formula for $w_i$. The condition $w \bot V$ is equivalent to
$$\sum_{i=0}^{n-k+1} w_i x_i^j = 0, \quad \forall j = 0, 1, \dots, n-k$$
This means that $\sum_{i=0}^{n-k+1} w_i p(x_i) = 0$ for all polynomials $p$ of degree at most $n-k$. Define $w$ in the following way:
$$w_i = \prod_{\substack{j=0 \\ j \neq i}}^{n-k+1} \frac{1}{x_i - x_j}, \quad i = 0, 1, \dots, n-k+1$$
These are the barycentric weights associated with Lagrange interpolation at the points $(x_i)$. 
That is, the Lagrange basis for  $(x_i)$ consists of polynomials:
$$l_i = \prod_{\substack{j=0 \\ j \neq i}}^{n-k+1} \frac{x-x_j}{x_i - x_j}, \quad i = 0, 1, \dots, n-k+1$$
with:
$$p(x) = \sum_{i=0}^{n-k+1} p(x_i)l_i(x)$$
Since the degree of $p(x)$ is $n-k$, it must be the case that in the above expression the coefficient at $x^{n-k+1}$ equals zero, which translates to:
$$\sum_{i=0}^{n-k+1} \prod_{\substack{j=0 \\ j \neq i}}^{n-k+1} \frac{p(x_i)}{x_i - x_j} = 0$$
and so $\langle w, v \rangle = 0$, which completes the proof.
\end{remark}
Now, we are ready to show the reduction from the Partition Problem. Let 
$$\sum_{i=0}^{n-k} b_i t_i  = 0$$
be an instance of the Partition Problem with integers $t_i$. Let us set $t_{n-k+1} = \frac{1}{3}$. 
For fixed $\tuple{\overline{r}, \overline{v}}$ and $x_i$, define $y_i = 3|t_i\frac{S}{\alpha_i}|$. Then the solution to the interpolation problem has the following form:
$$3\sum_{i=0}^{n-k} b_i t_i + b_{n-k+1}  = 1$$
We claim that solutions to the above problem are tantamount to solutions of the Partition Problem. Indeed, if $b_i$ are the solution to the above interpolation problem, then it must be that $b_{n-k+1} = 1$ and then:
$$3\sum_{i=0}^{n-k} b_i t_i + 1 = 1 \Leftrightarrow \sum_{i=0}^{n-k} b_i t_i = 0$$
Otherwise, i.e.~if $b_{n-k+1} = -1$, we have:
$$3\sum_{i=0}^{n-k} b_i t_i - 1 = 1 \Leftrightarrow \sum_{i=0}^{n-k} b_i t_i = \frac{2}{3}$$
which is impossible, because $t_i$, for $0 \leq i \leq n-k$, are integers.
Since $k = cn$ with $0 \leq c < 1$, we have that $n-k = n-cn=(1-c)n = \Theta(n)$. Therefore, our transformation is polynomial in bit-size.

\end{proof}

\begin{theorem}[Phaseless retrieval from $(1+c) n + 2$ points with $0 \geq c < 1$ is NPC]
Let $r_0, r_1, r_2, \dots$ be an infinite sequences of pairwise-distinct rational numbers (the evaluation nodes), such that the bit-size of $r_i$ is polynomial in $i$.  The problem of identifying a polynomial of degree $n$ up to its phase from non-adaptive $n+2$ phaseless evaluations and $cn$ exact evaluations in nodes $r_0, r_1, \dots, r_{n+\lfloor cn \rfloor}$ is NP-complete for any fixed $0 \leq c < 1$.
\end{theorem}
\begin{proof}
Let us evaluate our polynomial at nodes $r_0, r_1, \dots, r_{\lfloor cn \rfloor}$ exactly. This yields non-negative rational points $v_0, v_1, , \dots, v_{\lfloor cn \rfloor}$. The reduction from phaseless polynomial retrieval over $\tuple{\overline{r}, \overline{v}}$ is trivial, because solutions $p \in \mathbf{Q}_n[x]$ to the above problem with fixed nodes $\overline{r}$ of exact evaluations are tantamount to solutions for phaseless polynomial retrieval for fixed $\tuple{\overline{r}, \overline{v}}$, where $\overline{v}$ are corresponding evaluation values.
\end{proof}

\begin{remark}
    The above can be generalised for any $k$ such that $n-k = O(n^p)$ for any $p>0$. Therefore, setting $k = cn$ for any $c < 1$ gives the result mentioned in the introduction, because we have: $n-cn = (1-c)n = \Theta(n)$. Nonetheless, we can also get a stronger result. Let $k = n - n^p$ for some $0 < p < 1$, then: $n-k = n - n + n^p = n^p = \Theta(n^p)$. Therefore, the problem of identifying a polynomial of degree $n$ from $2n-n^p$ evaluations is NP-complete for any fixed $0 < p < 1$.
\end{remark}

\section*{Appendix}

\begin{lstlisting}[language=Python, caption={Phaseless Interpolation over $\Q$ with $2n+1-k$ points}]
import sympy as sp
from sympy.abc import x


# ==========================================
# Part 1: Triangular System Solver
# ==========================================

def _simplify_polys(polys, partial_solution):
    """Substitutes partial solutions and expands polynomials."""
    current_polys = [sp.expand(p.subs(partial_solution)) for p in polys]
    return [p for p in current_polys if p != 0]


def _check_inconsistency(active_polys):
    """Returns True if any polynomial simplifies to a non-zero number."""
    return any(p.is_Number and p != 0 for p in active_polys)


def _find_target_poly(active_polys, target_var):
    """Identifies a univariate polynomial for the target variable."""
    for p in active_polys:
        syms = p.free_symbols
        if not syms: continue
        if syms == {target_var} or syms.issubset({target_var}):
            return p
    return None


def _find_rational_roots(poly, target_var):
    """Solves a univariate polynomial and returns strictly rational roots."""
    roots = sp.solve(poly, target_var, dict=True)
    valid_roots = []
    for root in roots:
        val = root[target_var]
        if val.is_number and val.is_rational:
            valid_roots.append(val)
    return valid_roots


def solve_triangular_system(polys, variables, partial_solution=None):
    """Recursively solves a triangular system of polynomials."""
    if partial_solution is None: partial_solution = {}
    if not variables: return [partial_solution]

    target_var, remaining = variables[-1], variables[:-1]
    active_polys = _simplify_polys(polys, partial_solution)

    if _check_inconsistency(active_polys): return []

    uni_poly = _find_target_poly(active_polys, target_var)
    if uni_poly is None: return []

    full_solutions = []
    for val in _find_rational_roots(uni_poly, target_var):
        new_partial = partial_solution.copy()
        new_partial[target_var] = val
        full_solutions.extend(solve_triangular_system(polys, remaining, new_partial))

    return full_solutions


# ==========================================
# Part 2: Core Logic (_solve_core)
# ==========================================

def _setup_shifted_points(points, k, shift_val):
    """Applies coordinate shift and calculates degrees."""
    shifted = [(p[0] + shift_val, p[1]) for p in points]
    m = len(shifted)
    d = (m + k - 1) // 2
    return shifted, m, d


def _compute_p_coeffs(x_vals, y_vals, k, m, c_vars):
    """Computes coefficients of P(x; c) via interpolation."""
    L_expr = sp.interpolating_poly(m, x, x_vals, y_vals)
    R_expr = sp.prod([(x - xi) for xi in x_vals])
    S_expr = sum(c_vars[j] * x ** j for j in range(k))

    p_poly = sp.Poly(sp.expand(L_expr + S_expr * R_expr), x)
    max_deg = m + k - 1
    return {i: p_poly.coeff_monomial(x ** i) for i in range(max_deg + 1)}


def _compute_convolution_numerator(range_indices, a_terms, a0):
    """Computes the sum of a_j * a_{r-j} with manual exponent handling."""
    sum_num, sum_max_exp = 0, 0
    term_exps = [a_terms[j][1] + a_terms[range_indices.stop - 1 - j + range_indices.start][1]
                 for j in range_indices]

    if term_exps: sum_max_exp = max(term_exps)

    for idx, j in enumerate(range_indices):
        r_idx = range_indices.stop - 1 - j + range_indices.start
        num = a_terms[j][0] * a_terms[r_idx][0]
        diff = sum_max_exp - term_exps[idx]
        if diff > 0: num *= (2 * a0) ** diff
        sum_num += num

    return sum_num, sum_max_exp


def _compute_a_terms_recursive(d, p_coeffs, a0):
    """Generates a_r terms where a_r = numerator / (2*a0)^exp."""
    a_terms = {0: (a0, 0)}
    for r in range(1, d + 1):
        sum_num, max_exp = _compute_convolution_numerator(range(1, r), a_terms, a0)

        # Formula: a_r = (P_r * (2a0)^max_exp - sum_num) / (2a0)^(max_exp + 1)
        P_r = p_coeffs.get(r, 0)
        new_num = P_r * (2 * a0) ** max_exp - sum_num
        a_terms[r] = (sp.expand(new_num), max_exp + 1)
    return a_terms


def _build_gb_equations(d, m, k, p_coeffs, a_terms, a0):
    """Builds the system of equations for the Grobner basis."""
    # Initial equation: a_0^2 - P_0 = 0
    eqs = [sp.expand(a0 ** 2 - p_coeffs.get(0, 0))]

    max_deg = m + k - 1
    for r in range(d + 1, max_deg + 1):
        start_j, end_j = max(0, r - d), min(r, d)
        if start_j > end_j: continue

        q2_num, q2_exp = _compute_convolution_numerator(range(start_j, end_j + 1), a_terms, a0)
        # Eq: q^2_coeff - P_r = 0  =>  q2_num - P_r * (2a0)^exp = 0
        P_r = p_coeffs.get(r, 0)
        eqs.append(sp.expand(q2_num - P_r * (2 * a0) ** q2_exp))
    return eqs


def _reconstruct_poly_from_sol(sol, a_terms, d, shift_val, a0):
    """Reconstructs a single valid polynomial from a numeric solution."""
    #if sol[a0] == 0: return None

    print(f"Sol = {sol}")
    coeffs = {}
    for r in range(d + 1):
        num_poly, exp_val = a_terms.get(r, (0, 0))
        try:
            #val = num_poly.subs(sol) / ((2 * sol[a0]) ** exp_val)
            val = num_poly.subs(sol) / ((2 * a0) ** exp_val)
            if not (val.is_number and val.is_rational): return None
            coeffs[r] = val
        except ZeroDivisionError:
            return None

    Q_t = sum(coeffs[i] * x ** i for i in range(d + 1))
    return Q_t.subs(x, x + shift_val)


def _solve_core(points, k, shift):
    """Orchestrator for the core algebraic solving logic."""
    pts, m, d = _setup_shifted_points(points, k, shift[0])

    # Setup variables
    # shifted polynomial has a0^2 = y_i, where i is the shift value
    a0 = sp.sqrt(shift[1]) # or a0 = -sp.sqrt(shift[1])
    c_vars = [sp.Symbol(f'c_{k - i - 1}') for i in range(k)]
    gb_vars = c_vars

    # Compute P coeffs and recursive a_terms
    p_coeffs = _compute_p_coeffs([p[0] for p in pts], [p[1] for p in pts], k, m, c_vars)
    a_terms = _compute_a_terms_recursive(d, p_coeffs, a0)
    valid_polys = []
    if k > 0:
        # Build and solve system
        sys_eqs = _build_gb_equations(d, m, k, p_coeffs, a_terms, a0)
        try:
            print(sys_eqs)
            gb = sp.groebner(sys_eqs, gb_vars, order='lex', domain='QQ')
        except:
            return []

        print(f'gb = {gb}')
        if list(gb) == [1]: return []  # No solution
        solutions = solve_triangular_system(list(gb), gb_vars)
        valid_polys = [_reconstruct_poly_from_sol(s, a_terms, d, shift[0], a0) for s in solutions]
    else:
        valid_polys = [_reconstruct_poly_from_sol({}, a_terms, d, shift[0], a0)]
    return [p for p in valid_polys if p is not None]



# ==========================================
# Part 3: Phaseless Interpolation
# ==========================================


def _calculate_shift(points):
    for x, y in points:
        if y != 0: return -x, y
    return None


def _deduplicate_solutions(candidates):
    """Simplifies polynomials and filters out duplicates."""
    unique_polys = []
    seen_exprs = set()

    for p in candidates:
        simp_p = sp.simplify(p)
        if simp_p not in seen_exprs:
            unique_polys.append(simp_p)
            seen_exprs.add(simp_p)

    return unique_polys


def _log_final_results(solutions):
    """Prints the formatted final solutions to the console."""
    print(f"--- Results ---")
    for i, poly in enumerate(solutions):
        print(f"Solution {i + 1}: q(x) = {poly}")


def _solve_affine_square_roots(points, k):
    """
    Main driver: Orchestrates the search for affine square roots.
    1. Configures shift limits.
    2. searches for candidates.
    3. Deduplicates and logs results.
    """
    print(f"--- Configuration: Points={len(points)}, k={k} ---")

    shift = _calculate_shift(points)
    if shift is None:
        solutions = [0]
        print(f"No shift: it is the zero polynomial")
    else:
        solutions = _solve_core(points, k, shift=shift)
        print(f"Shift {shift}: Found {len(solutions)} solutions")

    unique_solutions = _deduplicate_solutions(solutions)
    _log_final_results(unique_solutions)

    return unique_solutions


def phaseless_interpolation(points, k):
    squared_points = [(x, y**2) for x, y in points]
    return _solve_affine_square_roots(squared_points, k)


if __name__ == "__main__":
    # Test Case 1: y = x
    print("Test Case 1: y = x")
    points1 = [(0, 0), (1, -1), (2, 2)]
    res1 = phaseless_interpolation(points1, k=0)
    print("\n")

    # Test Case 2: y = (x+1)
    print("Test Case 2: y = (x+1)")
    points2 = [(0, 1), (1, -2), (2, 3)]
    res2 = phaseless_interpolation(points2, k=0)
    print("\n")

    # Test Case 3: y = (x+1)^2
    print("Test Case 3: y = (x+1)^2, k=1, but wrong evaluation at x=-2")
    points3 = [(-2, 9), (0, 1), (1, 4), (2, 9)]
    res3 = phaseless_interpolation(points3, k=1)
    print("\n")

    # Test Case 4: Higher degree
    print("Test Case 4: y = x^5 - 6x^4 + 5x^3 + 4x^2 - 3x + 2 from 10 points")
    points4 = [(1, 3), (2, 12), (3, 79), (4, 138), (5, 87), (6, 1208), (-1, 3), (-2, 144), (-3, 817), (0, 2)]
    res4 = phaseless_interpolation(points4, k=1)

    # Test Case 5: Higher degree, higher degree of freedom
    print("Test Case 5: y = x^5 - 6x^4 + 5x^3 + 4x^2 - 3x + 2 from 9 points")
    points5 = [(1, 3), (2, 12), (3, 79), (4, 138), (5, 87), (6, 1208), (-1, 3), (-2, 144), (-3, 817)]
    res5 = phaseless_interpolation(points5, k=2)

    # Test Case 6: Higher degree, higher degree of freedom
    print("Test Case 5: y = x^5 - 6x^4 + 5x^3 + 4x^2 - 3x + 2 from 8 points")
    points6 = [(1, 3), (2, 12), (3, 79), (4, 138), (5, 87), (-1, 3), (-2, 144), (0, 2)]
    res6 = phaseless_interpolation(points6, k=3)
\end{lstlisting}

\bibliography{bibliography}

\end{document}